\documentclass[journal]{new-aiaa}
\usepackage[utf8]{inputenc}

\usepackage{graphicx}
\usepackage{amsmath,ntheorem,mathtools}
\usepackage{longtable,tabularx}
\usepackage{algorithm}
\usepackage{algpseudocode}
\usepackage{subfigure}
\usepackage{booktabs}

\usepackage{epstopdf}
\usepackage{todonotes}

\newtheorem{theorem}{Theorem}

\newtheorem{assumption}{Assumption}[section]

\newtheorem{remark}{Remark}[section]
\newenvironment{proof}[1][Proof]{\begin{trivlist}
\item[\hskip \labelsep {\bfseries #1}]}{\end{trivlist}}
\newcommand{\qed}{\nobreak \ifvmode \relax \else
      \ifdim\lastskip<1.5em \hskip-\lastskip
      \hskip1.5em plus0em minus0.5em \fi \nobreak
      \vrule height0.75em width0.5em depth0.25em\fi}

\newcommand{\R}{\mathbb{R}}
\newcommand{\N}{\mathbb{N}}
\DeclareMathOperator*{\argmin}{arg\,min}

\newcommand{\xmax}{\overline{x}}
\newcommand{\xmin}{\underline{x}}
\newcommand{\ymax}{\overline{y}}
\newcommand{\ymin}{\underline{y}}
\newcommand{\zmax}{\overline{z}}
\newcommand{\zmin}{\underline{z}}
\newcommand{\adset}{\ensuremath{\mathcal S_{D}}}
\newcommand{\adsetxz}{\ensuremath{\mathcal S_{D_{xz}}}}
\newcommand{\adsety}{\ensuremath{\mathcal S_{D_{y}}}}

\newcommand{\I}{\mathbb{I}}

\newcommand{\figref}[1]{Fig.~\ref{#1}}
\newtheorem{rem}{Remark}

\title{Event-based Impulsive Control for Spacecraft Rendezvous Hovering Phases}

\author{Julio C. Sanchez\footnote{Ph.D. Candidate, Department of Aerospace Engineering, Higher School of Engineering; jsanchezm@us.es.}}
\affil{University of Seville, 41092, Seville, Spain}
\author{Christophe Louembet\footnote{Associate professor, Methods and Algorithms in Control; currently University of Toulouse, Paul Sabatier University, Laboratory of Analysis and Architecture of Systems; louembet@laas.fr.}}
\affil{LAAS-CNRS, University of Toulouse, CNRS, 31400, Toulouse, France}
\author{Francisco Gavilan \footnote{Assistant professor, Department of Aerospace Engineering, Higher School of Engineering; fgavilan@us.es} and Rafael Vazquez \footnote{Associate professor, Department of Aerospace Engineering, Higher School of Engineering; rvazquez1@us.es}}
\affil{University of Seville, 41092, Seville, Spain}
\begin{document}
	
\maketitle

\begin{abstract}
This work presents an event-triggered controller for spacecraft rendezvous hovering phases. The goal is to maintain the chaser within a bounded region with respect to the target. The main assumption is that the chaser vehicle has impulsive thrusters. These are assumed to be orientable at any direction and are constrained by dead-zone and saturation bounds. The event-based controller relies on trigger rules deciding when a suitable control law is applied. The prime control law consists on a single impulse, therefore the trigger rules design is based on the instantaneous reachability to the admissible set. The final outcome is a very efficient algorithm from both computational burden and footprint perspectives. Since the proposed methodology is based on a single impulse control, the controller invariance is local and assessed through  impulsive systems theory. Finally, numerical results are shown and discussed.
\end{abstract}

\section*{Nomenclature}

{\renewcommand\arraystretch{1.0}
	\noindent\begin{longtable*}{@{}l @{\quad=\quad} l@{}}
		$a$ & semi-major axis\\
		$A,~B$ & state and control matrices for Cartesian relative dynamics \\
		$A_D,~B_D$ & state and control matrices for relative orbit element dynamics \\
		$D$ & relative orbit element state \\
		$\mathcal{D}$ & admissible set region of attraction\\
		$\mathcal{D}_{\text{dz}}$ & dead-zone set\\
		$e$ & eccentricity \\
		$\mathcal{F}_{xz},~\mathcal{F}_y$ & state increment reachable set\\
		$G_{xz},~G_y$ & instantaneous reachability indicators\\
		$g_w(\cdot)$ & multivariate polynomials in $D$ states\\
		$L(\cdot)$ & length of a given interval\\
		$\adset$ & admissible set \\
		$t$ & time \\
		$U,~V$ & similarity transformation matrices\\
		$\xmin,~\xmax,~\ymin,~\ymax,~\zmin,~\zmax$ & polytopic constraints bounds\\
		$\textit{\textbf{X}}$ & Cartesian relative state \\
		$\mathcal Z$ & jump set\\
		$\delta_{xz},~\delta_y$ & trigger rules threshold\\
		$\Delta D$ & relative orbit element state increment\\
		$\Delta V$ & impulse\\
		$\underline{\Delta V}$ & dead-zone value\\
		$\overline{\Delta V}$ & saturation value\\
		$\Delta^+$ & instantaneous reachable set\\
		$\Delta^+_{\text{sat}}$ & instantaneous reachable set accounting for control limitations\\
		$\Delta^S$ & intersection set between admissible set and instantaneous reachable set\\
		$\lambda_{xz},~\lambda_y$ & control variables\\
		$\Lambda^S$ & set of admissible control variables with respect to \adset{}\\
		$\Lambda_{\text{sat}}$ & set of admissible control variables with respect to control limitations\\
		$\Lambda_{\text{sat}}^S$ & set of admissible control variables with respect to \adset{} and control limitations\\
		$\mu$ & Earth gravitational constant\\
		$\nu$ & true anomaly \\
		$\Phi(\nu,\nu_0)$ & state transition matrix from $\nu_0$ to $\nu$\\
		\multicolumn{2}{@{}l}{Subscripts}\\
		$xz$ & indicates in-plane motion\\
		$y$ & indicates out-of-plane motion\\
		$\text{dz}$ & indicates dead-zone constraint\\
		$\text{sat}$ & indicates dead-zone and saturation constraints\\
		\multicolumn{2}{@{}l}{Supercripts}\\
		+ & indicates a state after control
\end{longtable*}}

\section{Introduction}

Spacecraft rendezvous operations have played a key role in the space exploration, see \cite{Woffinden2007} for an historical review. In fact, new mission concepts have arisen with the increasing popularity of CubeSats in both industry and academia, see \cite{Shkolnik2018}. For example, the novel mission² concept proposed by \cite{Underwood2015}, plans to design an autonomous on-orbit reconfigurable telescope. Thus, the ability of a spacecraft to compute control commands on-board is a key component for novel rendezvous operations. 

This paper is focused on the hovering phase of rendezvous mission, see \cite{Irvin2009}. 
During the rendezvous mission, this phase permits a pause useful for observing the target satellite  or waiting for order to continue the mission in safe conditions.
A critical feature of the hovering phase to be ensured by the control algorithm is the safety of the mission that mainly consists in avoiding collision while waiting in the vicinity of the leader spacecraft.
To achieve such goals, the chaser's relative position should be maintained within a restricted zone away from the target with a minimal fuel consumption. This hovering region has to be defined in a local frame attached to the leader position while it should exclude the leader spacecraft for safety reasons.

Regarding propulsion devices, chemical thusters (for large and heavy spacecraft) and cold gas thrusters (for lightweight spacecraft) are being employed nowadays for geocentric space proximity operations, \cite{DiMauro2018}. In both cases, the control signal can be modelled in an impulsive way with adequate accuracy. Additionally, the impulse amplitude is constrained not only by saturation but by the minimum impulse bit (minimum force that the thruster can apply). This dead-zone constraint makes the rendezvous planning problem non-convex and difficult to tackle by conventional methods, see \cite{Hartley2013_bis}. A vast literature on impulsive rendezvous exists. References \cite{Breger2008, Yazhong2014} designed passive collision avoidance trajectories, \cite{Yazhong2007} studied the trade-off of several rendezvous performance indexes (fuel consumption, flight time and safety), \cite{Yang2015} developed a switching control strategy guaranteeing close-loop stability and \cite{Sanchez2019} designed a six-degrees of freedom controller based on flatness theory. Some relevant works on impulsive control can also be found in the field of spacecraft flight formation, see \cite{Qi2012, Gaias2015}. 

Regarding the hovering phase, strategies such as the "pogo" or "teardrop" approaches have been proposed in \cite{Irvin2008, Irvin2009}. These approaches are fuel inefficient as the satellite remains in the hovering zone by means of a systematic impulsive control which implements a bouncing-on-the-zone-frontiers strategy. On the contrary, in \cite{Arantes2019} the relative orbits periodicity property is exploited to maintain the chaser spacecraft in the  hovering zone without control effort in absence of disturbances.
In this aim, the relative motion invariants are used as parameters of the relative orbits in the framework of linearized relative motion accordingly with \cite{Georgia2013, Bennett2016}. Using \cite{Georgia2013} convenient parameterization,  \cite{Deaconu2012} employed polynomial positivity techniques to compute the coordinates of natural periodic relative orbits within a polytopic region. This methodology has also been extended to non-periodic relative orbits, see \cite{Deaconu2015}. However, the LMI conditions modelling the set membership constraints for a relative orbit have been demonstrated to be numerically cumbersome for spacecraft on-board computation devices. For this reason, in \cite{Arantes2017}, an implicitization method was employed to formally describe the admissible set of constrained orbits as polynomial inequalities in terms of the relative state. Finally, \cite{Arantes2019} designed a global stability controller, based on a three impulsive sequence, for the previously described admissible set.     

Rendezvous impulsive control has been studied in the last decade using different approaches such as the MPC framework (see \cite{Arantes2019} and references therein) or the hybrid systems framework \cite{Brentari2018}. 
The controllers from the literature deliver periodically the control impulse to be executed. Although stability has been demonstrated in such frameworks, it appears that most of the controls are \emph{filtered} by the thrusters limitations (the minimum impulse bit in particular). This leads to a slight degradation of the controller performances or even to critical situations closed to instability. In addition, it could be desirable to reduce control computations \cite{Heemels2012}.
Under this new requirement, the event-based control methodology has emerged as an alternative to classical periodic controllers. In this control paradigm, the commands are computed aperiodically, thus reducing communications between sensors, controllers and actuators, see \cite{Astrom2008} for the basics. The event-based methodology can be effectively combined with feedback policies, see \cite{Wu2014}, or MPC schemes as in \cite{Pawlowski2015}. More specifically, event-based control is gaining momentum amongst the spacecraft control community with  attitude control applications, see \cite{Wu2018, Zhang2018}.

The main purpose of this paper is to extend and complete the preliminary results presented in the conference papers of \cite{Louembet2018, Sanchez2019_bis} which aimed to overcome some drawbacks of the global stabilizing controller proposed by \cite{Arantes2019}. 
In particular, using \cite{Arantes2019} global controller, the set membership conditions are only ensured after the $N$-impulse maneuver ($N\ge 3$). Moreover, unnecessary impulses may be commanded when the spacecraft is close enough to the hovering region. To address these issues, an event-based predictive controller for spacecraft rendezvous hovering phases is designed in this work. Since the proposed controller is local, the main assumption is that the chaser vehicle is already in the hovering region vicinity. 

The key components of the proposed event-based predictive controller are, on the one hand, a set of rules triggering the control decision based on target set proximity and reachability criteria. On the other hand, the computation of a single control impulse is the second leg of this controller. An advantage of the resulting event-based algorithm is the low level of numerical complexity which leads to a reduced computation burden. As a matter of fact, target set proximity and reachability conditions rely on univariate polynomial roots computation whereas the single control impulse computation boils down to a low dimension linear program. Finally, using impulsive systems theory, see \cite{Haddad2006}, the primal single impulse control invariance, in a local way, is assessed.           

This paper is structured as follows. Section II introduces the impulsive relative motion model and the algebraic description of constrained orbits. Section III presents the event-based predictive controller composed by the control law and trigger rules. Section IV assess the invariance of the primal single impulse approach. Section V proposes and analyses numerical results of interest. Finally, Section VI concludes the paper with some remarks.

\section{Relative motion and constrained orbits}

In this section, the relative motion model between two spacecraft on Keplerian orbits is presented and parameterized in a convenient manner. Then, the set of relative constrained orbits is formally described by means of implicitization techniques. 

\subsection{Relative motion}

In orbit proximity operations the vehicles are considered close as the ratio of the relative distance with respect to the semi-major axis is very small. This fact allows to linearize the relative dynamics around the target orbit. If the relative dynamics model only includes Keplerian effects, the classic Tschauner-Hempel linear time-varying (LTV) system \cite{Tschauner1967} can be employed. For this particular LTV model, a formal fundamental matrix is described in the literature and a given transition matrix was proposed in \cite{Yamanaka2002}.

\paragraph{Relative dynamics and state transition}

The chaser spacecraft relative motion, denoted by $S_f$, is expressed with respect to the local frame attached to a passive target with its position denoted by $S_l$.
The local frame $\{ S_l,\vec{x},\vec{y},\vec{z} \}$, denoted as Local Vertical/Local Horizontal (LVLH), evolves around the Earth-centered inertial frame, $\{O, \vec{I},\vec{J},\vec{K} \}$, along the target spacecraft orbit. Note that $\vec{z}$ is the radial vector (positive towards the centre of the Earth), $\vec{y}$ is the cross-track vector (opposite to the orbit angular momentum) and $\vec{x}$ is the in-track vector completing a right-handed system, see Fig.\ref{fig:LVLH_frame}.

\begin{figure}[h] 
	\begin{center}
		\includegraphics[width=9.5cm,height=9.5cm,keepaspectratio]{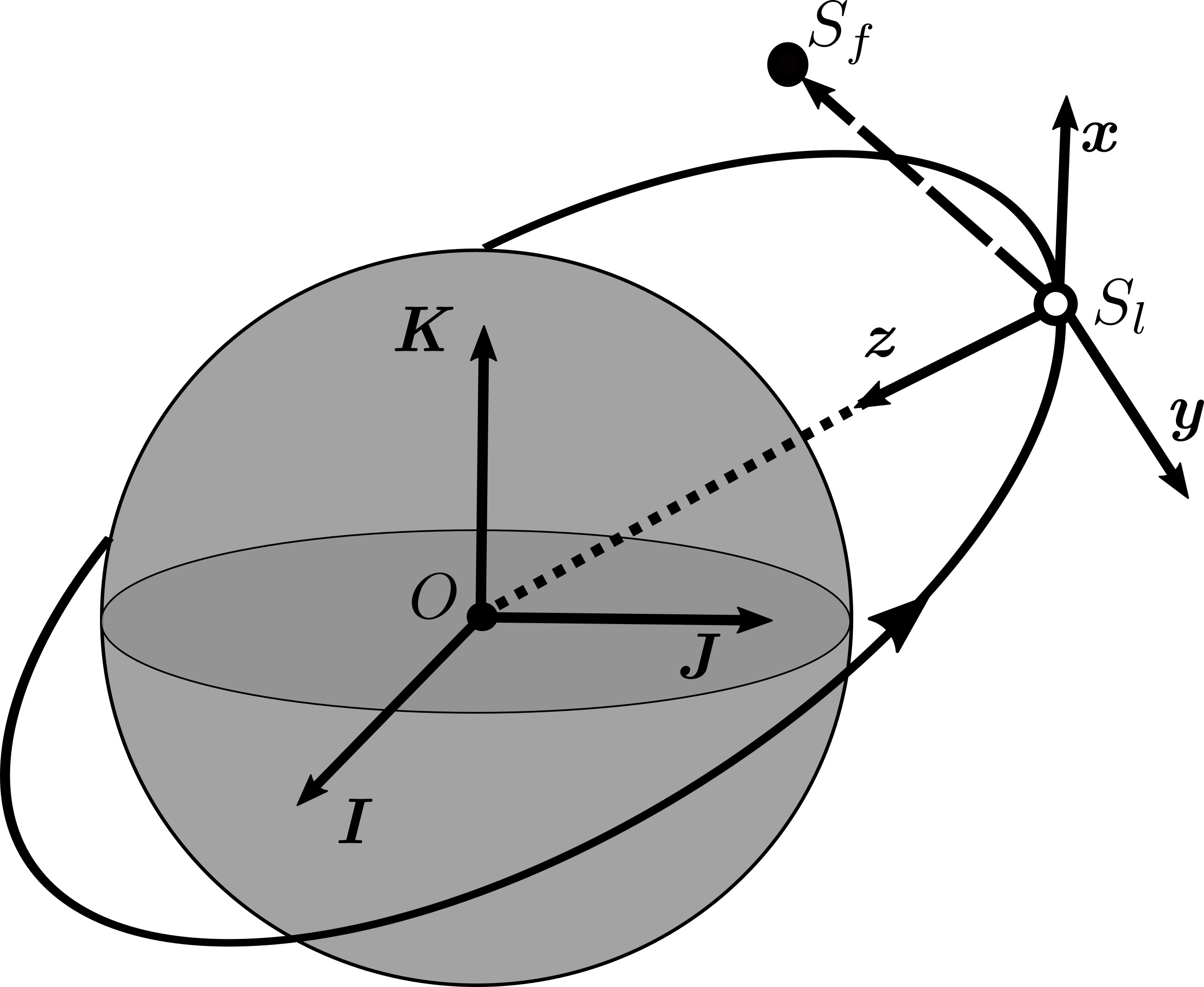}
	\end{center}
	\caption{Inertial Earth-centered and LVLH frames.}	
	\label{fig:LVLH_frame}
\end{figure}

Under Keplerian assumptions, the relative motion between two spacecrafts can be expressed by means of the Tschauner-Hempel equations, see \cite{Tschauner1967}. Considering the vehicles are close enough, i.e. $ \| \overrightarrow{OS_l}\|_2 \gg \|\overrightarrow{S_lS_f}\|_2$, these equations can be linearized to obtain the following linear time-varying dynamics
\begin{equation}\label{DynLinT}
\dot X(t)=A(t)X(t),
\end{equation} 
where the state vector $X$ represents the relative position and velocity in the LVLH frame $X(t)=[x(t),~y(t),~z(t),~\dot x(t),~\dot y(t),~\dot z(t)]^T$. In this work, the transition matrix of Eq.\eqref{DynLinT} dynamics is exploited. To obtain this matrix, a similarity transformation is applied
\begin{equation}\label{eq:space2tilde}
\tilde X(\nu)= U(\nu) X(t), \> \> \text{with} \> \> U(\nu) = \begin{bmatrix}
\rho\mathbb{I}_3 & 0_3\\
\rho' \, \mathbb{I}_3  & \left(k^2 \rho\right)^{-1} \mathbb{I}_3
\end{bmatrix},
\end{equation} where $(\cdot)'=\frac{d(\cdot)}{d\nu}$, $k^2=\sqrt{\frac{\mu}{a^3(1-e^2)^3}}$, $\rho=1+e\cos\nu$ and $\mathbb{I}_3$ denotes the identity matrix. The variables $a$ and $e$ are the target orbit semi-major axis and eccentricity respectively. The similarity transformation of Eq.\eqref{eq:space2tilde} is a change of the independent variable from time $t$ to true anomaly of the target spacecraft, $\nu$, which defines its position through the orbit. In this framework, the transition matrix can be analytically obtained (see \cite{Yamanaka2002}) so that
\begin{equation}
\tilde{X}(\nu)=\Phi(\nu,\nu_0)\tilde{X}(\nu_0),\quad \nu_0\le\nu.
\end{equation}

\paragraph{Parameterizing the relative motion}

In \cite{Deaconu2012}, the relative position has been explicitly expressed in a convenient manner as
\begin{equation}\label{position_transition}
\begin{aligned}
\tilde{x}(\nu)&=d_1(1+\rho)s_\nu-d_2(1+\rho)c_\nu+d_3+3d_0J(\nu)\rho^2,\\
\tilde{y}(\nu)&=d_4c_\nu+d_5s_\nu,\\
\tilde{z}(\nu)&=d_1\rho c_\nu+d_2\rho s_\nu-3ed_0J(\nu)s_\nu \rho+2d_0,\\
\end{aligned}
\end{equation}
where $s_{\nu}=\sin{\nu}$, $c_{\nu}=\cos{\nu}$ and $J(\nu)$ is given by
\begin{equation}\label{eq:J}
J(\nu) := \int_{\nu_0}^{\nu} \frac{d\tau}{\rho(\tau)^2} = \sqrt{\frac{\mu}{a^3}}\frac{t-t_0}{(1-e^2)^{3/2}}.
\end{equation}
As the classic orbital elements, the parameters $d_i~(i=0\hdots5)$ describe the chaser relative orbit in terms of its instantaneous shape and position, see \cite[chap. 2]{Georgia2013}. This fact makes the vector $D=[d_0,~d_1,~d_2,~d_3,~d_4,~d_5]^T$ a relevant state when aiming to constrain relative orbits. Note that a linear transformation, between the relative state $\tilde{X}$ and the vector $D$, exists as
\begin{equation}\label{D2X}
\tilde{X}(\nu)=V(\nu)D(\nu),
\end{equation}
with 
\begin{equation*}\label{eq:def_D_mat}
V(\nu) = \begin{bmatrix}
0 & s_{\nu}(1+\rho) & -c_{\nu}(1+\rho) & 1 & 0 & 0\\
0 & 0 & 0 & 0 & c_{\nu} & s_{\nu}\\
2 & c_{\nu}\rho & s_{\nu}\rho & 0 & 0 & 0\\
3 & 2c_{\nu}\rho-e & 2s_{\nu}\rho & 0 & 0 & 0\\
0 & 0 & 0 & 0 & -s_{\nu} & c_{\nu}\\
-\dfrac{3es_{\nu}}{\rho} & -s_{\nu}(1+\rho) & 2ec_{\nu}^2-e+c_{\nu} & 0 & 0 & 0\\
\end{bmatrix}.
\end{equation*}
Since $\det V = 1-e^2 \neq 0, ~\forall e \in  [0, 1)$, Eq.\eqref{D2X} represents a similarity transformation and $D$ is a proper state vector with its own dynamics governed by ${A_D(\nu)}$
\begin{equation}\label{eq:D_dyn}
D'(\nu)=\underbrace{
	\begin{bmatrix}
	0 & 0 & 0 & 0 & 0 & 0\\
	0 & 0 & 0 & 0 & 0 & 0\\[1 em]
	-\tfrac{3e}{\rho^{2}} & 0 & 0 & 0 & 0 & 0\\[1 em]
	\tfrac{3}{\rho^{2}} & 0 & 0 & 0 & 0 & 0\\
	0 & 0 & 0 & 0 & 0 & 0\\
	0 & 0 & 0 & 0 & 0 & 0\\
	\end{bmatrix}}_{A_D(\nu)}
D(\nu),
\end{equation}
and its own transition matrix $\Phi_{D}(\nu,\nu_0)$,
\begin{equation}\label{eq:D_trans_mat}
D(\nu)= \underbrace{ 
	\begin{bmatrix}
	1 & 0 & 0 & 0 & 0 & 0\\
	0 & 1 & 0 & 0 & 0 & 0\\
	-3eJ(\nu) & 0 & 1 & 0& 0 & 0\\
	3J(\nu) & 0 & 0 & 1 & 0 & 0\\
	0 & 0 & 0 & 0 & 1 & 0 \\
	0 & 0 & 0 & 0 & 0 & 1 \\
	\end{bmatrix}}_{\Phi_{D}(\nu,\nu_0)}
D(\nu_0).
\end{equation}

\paragraph{Impulsive control}
Typically, for space hovering operations, the chaser spacecraft is controlled by chemical engines providing a high level of thrust during a short period of time with respect to the target orbit period. In practice, this fact leads to an extremely fast velocity change which can be modelled in an impulsive way
\begin{equation}
X^+(t)=X(t)+B \Delta V(t), \quad B=[0_3,~\I_3]^T,
\end{equation} 
where $0_3$ is the square null matrix. Applying the changes of variable given by Eq.\eqref{eq:space2tilde} and Eq.\eqref{D2X}, an impulse at instant $\nu$ produces an instantaneous change in the state $D$ as
\begin{equation}\label{eq:jumpD}
D^+(\nu)=D(\nu)+B_D(\nu) \Delta V(\nu),
\end{equation}
with
\begin{equation}\label{eq:B_D}
B_D(\nu)=V^{-1}(\nu)U(\nu)B,
\end{equation}
which can be further developed as
\begin{equation}\label{eq:def_BD_mat}
B_D(\nu) = \frac{1}{k^2(e^2-1)\rho}
\begin{bmatrix}
\rho^2 & 0 & -es_{\nu}\rho\\
-2c_{\nu}-e(1+c_{\nu}^2) & 0 & s_{\nu}\rho\\
-s_{\nu}(2+ec_{\nu}) & 0 & 2e-c_{\nu}\rho\\
es_{\nu}(2+ec_{\nu}) & 0 & ec_{\nu}\rho-2\\
0 & -(e^2-1)s_{\nu} & 0\\
0 & (e^2-1)c_{\nu} & 0\\
\end{bmatrix}.
\end{equation}
Equation \eqref{eq:jumpD} shows that a given impulsive control, $\Delta V$, will have a different impact (in the $D$ space) depending on the application instant, $\nu$. This is due to the control matrix $B_D$ time dependence. Additionally, the impulse amplitude has to comply with the following conditions on dead-zone and saturation
\begin{equation}\label{sat_dz_conds}
\underline{\Delta V} \leq \lVert \Delta V \rVert_2 \leq \overline{\Delta V},
\end{equation}
where it is assumed that the thruster can be pointed at any direction.
  
\paragraph{Decoupling the relative motion}
Observing Eq.\eqref{position_transition}, one can notice that the in-plane, $xz$, and out-of-plane, $y$, relative motions are decoupled. These motions are represented by the state sub-vectors $D_{xz}=[d_0,~d_1,~d_2,~d_3]^T$ and $D_y=[d_4,~d_5]^T$ respectively. 
Consequently, the control matrix given by Eq.\eqref{eq:def_BD_mat} can be decomposed into two submatrices for both the in-plane and out-of-plane motions:
\begin{eqnarray}
	B_{D,xz}(\nu)&=&\frac{1}{k^2(e^2-1)\rho}
	\begin{bmatrix}
		\rho^2 & -es_\nu\rho\\
		-2c_\nu-e(1+c_\nu^2) & s_\nu\rho\\
		-s_\nu(2+ec_\nu) & 2e-c_\nu\rho\\
		es_\nu(2+ec_\nu) & ec_\nu\rho - 2
	\end{bmatrix},\label{eq:def_BD_xz}\\
	B_{D,y}(\nu)&=&\frac{1}{k^2(e^2-1)\rho}
	\begin{bmatrix}-(e^2-1)s_\nu\\(e^2-1)c_\nu\end{bmatrix}\label{eq:def_BD_y}.
\end{eqnarray}

\subsection{Constrained orbits}

In this work, the control objective is to maintain the spacecraft hovering inside a predefined polytopic subset of the relative position space. Thereafter, a cuboid is considered without loss of generality:
\begin{equation}\label{eq:constr1}
\xmin \leq  x(t) \leq \xmax, \quad \ymin \leq  y(t) \leq \ymax, \quad \zmin \leq  z(t) \leq \zmax, \quad \forall t \geq t_0. 
\end{equation}
If the target is a space vehicle (it could also just be a reference position), the cuboid $\{\xmin,\xmax,\ymin,\ymax,\zmin,\zmax\}$ should not contain the origin in order to avoid collisions. The most economic way to hover within a given region is that the chaser spacecraft evolves on naturally constrained periodic orbits. In \cite{Deaconu2012}, the periodic orbits are described by the necessary and sufficient periodicity condition, namely $d_0=0$. 
Note that this periodicity condition defines the set of equilibrium points in the state-space $D$ (cf. dynamic equation \eqref{eq:D_dyn}).
Inserting the changes of variables of Eq.\eqref{eq:space2tilde} and Eq.\eqref{D2X} into the polytopic constraints, given by Eq.\eqref{eq:constr1}, one obtains the constraints inequalities expressed in the $D$ space. Therefore, the admissible set $\adset$ can be formally described as
\begin{equation}\label{eq:admset}
\adset:=\left\lbrace
\begin{array}{c}
D \in \mathbb{R}^6
\end{array}
\, \middle| \,
\begin{array}{l}
d_0=0,
\end{array}
\,\,
\begin{array}{l}
\xmin \leq  V_x(\nu) D \leq \xmax \\
\ymin \leq  V_y(\nu) D \leq \ymax \\
\zmin \leq  V_z(\nu) D \leq \zmax 
\end{array}, \> \> \forall \nu \right\rbrace,
\end{equation}
where $V_x$, $V_y$ and $V_z$ are the first three rows of $V$, divided by $\rho$, respectively. Equation \eqref{eq:admset} states that the admissible set $\mathcal{S}_D$ represents of all the periodic relative orbits constrained within the hovering region given by Eq.\eqref{eq:constr1}. The reader can note that, as a subset of equilibrium set, $\adset$ is invariant by nature in absence of control. On figure \ref{fig:constrained_periodic_relative_orbits}, a few elements of $\adset{}$ are depicted.
\begin{figure}[h] 
	\begin{center}
		\includegraphics[width=12cm,height=12cm,keepaspectratio]{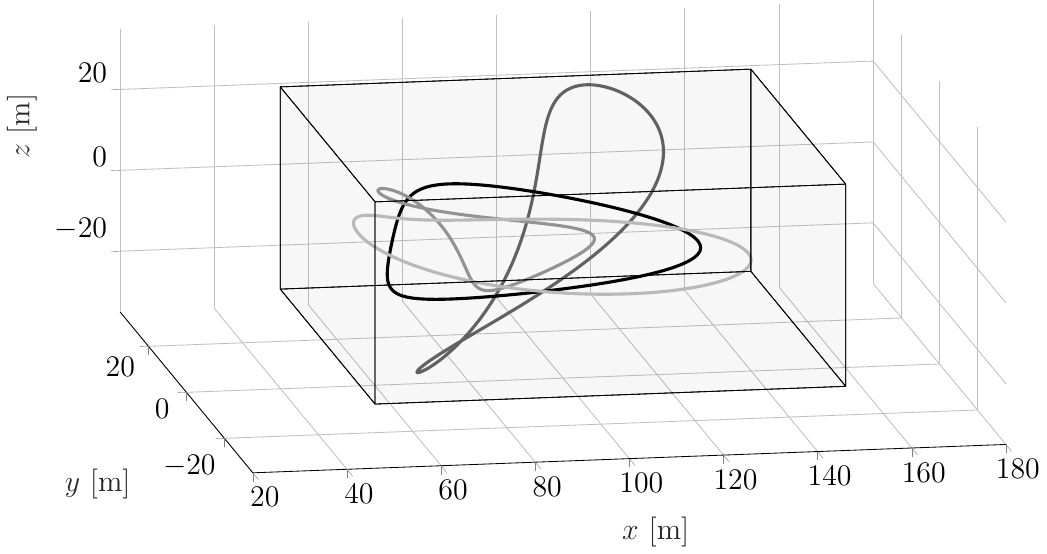}
	\end{center}
	\caption{Constrained periodic relative orbits.}	
	\label{fig:constrained_periodic_relative_orbits}
\end{figure}
The admissible set, $\adset$, is described by linear but time-varying conditions on the state $D$. In \cite{Arantes2017}, an implicitization method (see \cite{Hong1995}) is employed to obtain the envelope of each linear constraint composing the admissible set given by \eqref{eq:admset}. As a consequence, a semi-algebraic (free of the independent variable $\nu$) description of the admissible set is available as
\begin{equation}\label{eq:adset_envelope}
\adset=\{D \in \mathbb{R}^6 \mid d_0 = 0 \mid g_{w}(D) \leq 0,~\forall w \in \lbrace \xmin, \xmax, \ymin, \ymax, \zmin, \zmax\} \}.
\end{equation}
The admissible set $\adset$ has been demonstrated to be a closed convex set, see \cite{Arantes2019}. 
The functions $g_w(D)$ are multivariate polynomials in $d_1,~d_2,~d_3,~d_4$ and $d_5$
\begin{align}
g_{\xmin}(\mathbf{d}_{xz}, e, \xmin)& = \sum_{\beta\in\N_6^3} \underline\theta_\beta(e,\xmin)\mathbf{d}_{xz}^\beta\label{eq:xmin_envelope_interior},\\
g_{\xmax}(\mathbf{d}_{xz}, e, \xmax)& = \sum_{\beta\in\N_6^3} \overline\theta_\beta(e,\xmax)\mathbf{d}_{xz}^\beta,\label{eq:xmax_envelope_interior}\\
{g}_{\underline{y}}(d_4, d_5, e, \ymin)& = (d_4-e\underline{y})^2+d_5^2-\underline{y}^2, \label{eq:ymin_envelope}\\
{g}_{\overline{y}}(d_4, d_5, e, \ymax)& = (d_4-e\overline{y})^2+d_5^2-\overline{y}^2, \label{eq:ymax_envelope}\\
g_{\underline{z}}(d_1, d_2, \zmin)& = d_1^2+d_2^2-\underline{z}^2, \label{eq:zmin_envelope}\\                  
g_{\overline{z}}(d_1, d_2, \zmax)& = d_1^2+d_2^2-\overline{z}^2. \label{eq:zmax_envelope}                  
\end{align}
where $\mathbf d_{xz}=[d_1,~d_2,~d_3]^T$. The degree of the multivariate polynomials $g_{\xmin}(\cdot)$ and $g_{\xmax}(\cdot)$ is 6 for both cases. The set $\N_6^3$ is described as $\N_6^3:=\{\beta\in\N^3~:~\sum_{j=1}^3 \beta_j\leq 6 \}$; and its cardinality is $\binom{9}{6}$. Each term of these polynomials reads $\theta_{\beta_i} d_1^{\beta_{i,1}} d_2^{\beta_{i,2}} d_3^{\beta_{i,3}}$, $\beta_i\in\N_6^3$. For further details about the envelopes computation and admissible set boundaries characterization, please refer to \cite{Arantes2019}.
 
\section{Event-based predictive controller} \label{event_based_algorithm}

The proposed event-based predictive controller aims to compute and trigger a control impulse only when it is necessary in order to stabilize the admissible set $\adset$. Roughly speaking, the event-based philosophy consists in the following rules:\\
If the state belongs to the admissible set, no action is performed.\\
On the contrary, if the state does not belong to admissible set, a control decision is triggered based on the following conditions:
\begin{itemize}
	\item If the admissible set is reachable and certain reachability indicators (that are made clear in subsequent sections) fall below a threshold, the control acts.
	\item If the admissible set is currently unreachable but it is predicted to be reachable at any time during the next 2-$\pi$ time period, the controller will wait for that opportunity.
	\item If the previous conditions are not met, a back-up controller is used to ensure the stability. 
\end{itemize}

This section develops the different metrics to evaluate the admissible set reachability and proximity. Determining the admissible set reachability only makes sense when the spacecraft is outside of it, thus the trivial case of null control is not considered. Then, a proper set of trigger rules, based on the previous metrics, is developed. Finally, the single impulse optimal computation is presented and its numerical efficiency assessed.

Note that, as the in-plane and out-of-plane motions are decoupled, (see Eq.\eqref{eq:D_dyn} and Eq.\eqref{eq:B_D}), they are often treated separately in the rest of the paper.

\subsection{Reachability conditions}

The trigger rules design is based on the $\adset$ reachability conditions. To establish these conditions, the reachable set is described in the following sections. This reachable set is defined as the set of states that can be reached with a single control impulse, from a current state $D$, accounting for the thruster capabilities. Firstly, the instantaneous reachable set will be addressed. This set is the reachable set at the current time $\nu$. The instantaneous reachability condition is then derived by determining if there is an intersection between the instantaneous reachable set and the admissible set $\adset$. The 2-$\pi$ period reachable set is obtained by extending the instantaneous reachable set properties. Then, the 2-$\pi$ period reachability conditions are stated. Note that the latter set can be also seen as a region of attraction. 

\subsubsection{Instantaneous reachable set}
 This paragraph describes the instantaneous reachable set for both the out-of-plane and in-plane motions accounting for the thrusters saturation and minimum impulse bit:
\begin{equation}
\Delta^+_{\text{sat}}(D,\nu,\Delta V)=\Delta^+_{\text{sat},xz}(D,\nu,\Delta V)\times\Delta^+_{\text{sat},y}(D,\nu,\Delta V).
\end{equation}

\paragraph{Out-of-plane motion} 

The out-of-plane motion, represented by $D_y=[d_4,~d_5]^T$, is naturally periodic, thus an out-of-plane impulse $\Delta V_y=\lambda_y\in\mathbb{R}$ (the variable $\lambda_y$ is chosen to represent the out-of-plane impulse for notation consistency with the in-plane case) does not alter the orbit periodicity
\begin{equation}\label{eq:Dy+}
D^{+}_{y}(D_y, \nu, \lambda_y)=D_{y}+\lambda_yB_{D,y}(\nu),
\end{equation}   
where $B_{D,y}\in\mathbb{R}^{2}$ is given in Eq.\eqref{eq:def_BD_y}. The out-of-plane impulse allows the state $D_y$ to instantaneously change through the line $\Delta^+_{y}$ 
\begin{equation}\label{setDelta+y}
\Delta^+_{y}(D_y, \nu)=\{D^+_{y}(D_y,\nu,\lambda_y)\in\R^2 \text{ s.t. } \text{Eq.}\eqref{eq:Dy+}, \> \> \lambda_{y}\in\R  \}.
\end{equation}
However, the impulse amplitude is constrained due to dead-zone and saturations conditions. The set $\Lambda_{\text{sat},y}$ describes the out-of-plane thruster dead-zone and saturation conditions such that 
\begin{equation}\label{eq:Isat_y}
\Lambda_{\text{sat},y}=[-\overline{\Delta V},~-\underline{\Delta V}] \cup [\underline{\Delta V},~\overline{\Delta V}].
\end{equation}\label{Lambda_saty}
Note that $\Lambda_{\text{sat},y}$ does not depend on $\nu$. By using $\Lambda_{\text{sat},y}$, one can define the out-of-plane instantaneous reachable set
\begin{equation}\label{eq:Deltaysat}
\Delta^+_{\text{sat},y}(D_y,\nu)=\{D^+_{y}(D_y,\nu,\lambda_y)\in\R^2~~\text{s.t.}~~ \lambda_{y}\in \Lambda_{\text{sat},y}\},
\end{equation}
which is composed by two segments of the line $\Delta^+_{y}$.

\paragraph{In-plane motion}

The in-plane motion is described by the state subset $D_{xz}=[d_0,~d_1,~d_2,~ d_3]^T$. As stated previously in the literature \cite{Brentari2018,Arantes2019}, periodicity (also studied under Earth's oblateness effects \cite{Roscoe2011}) is a desirable property to hover around a specified region. Consequently, the in-plane control aims to steer the system to a constrained periodic orbit with a single in-plane impulse. After an in-plane impulse $\Delta V_{xz}=[\Delta V_x,~\Delta V_z]^T$, the in-plane state, $D^{+}_{xz}$ is given by
\begin{equation}\label{eq:Dxz+_first}
D^{+}_{xz}(D_{xz},\nu,\Delta V_{xz})=D_{xz}+B_{D,xz}(\nu)\Delta V_{xz},
\end{equation}
where $B_{D,xz}\in\mathbb{R}^{4\times2}$ is given by Eq.\eqref{eq:def_BD_xz}. To obtain a periodic orbit, the state $d_0$ is systematically steered to zero after each in-plane control impulse:
\begin{equation}\label{eq:d0+=0}
d^{+}_{0}(d_0,\nu,\Delta V_{xz}) = d_0+B_{d_0,xz}(\nu)\Delta V_{xz}=0,
\end{equation}
with $B_{d_i,xz}\in\mathbb{R}^{2}$, being the first row of $B_{D,xz}$. An impulse satisfying Eq.\eqref{eq:d0+=0} can be written in a general form as
\begin{equation}\label{eq:deltaV_xz}
\Delta V_{xz}(d_0, \nu,\lambda_{xz})=\lambda_{xz}B^\bot_{d_0,xz}(\nu)+\Delta V^{0}_{xz}(d_0,\nu),
\end{equation}
where $\lambda_{xz}\in\R$ is the in-plane control variable, $B^\bot_{d_0,xz}\in\mathbb{R}^2$ describes the kernel space of $B_{d_0,xz}$ (it is chosen as $||B^\bot_{d_0,xz}||_2=1$) and $\Delta V^{0}_{xz}\in\mathbb{R}^2$ is any particular solution of Eq.\eqref{eq:d0+=0}. Note that $\Delta V_{xz}^0$ always exists since the first entry of $B_D$ is $\rho/(e^2-1)k^2<0 ~\text{for}~ 0\leq e<1~\text{and}~\nu\in[0,~2\pi]$ . Using this periodicity-pursuing strategy, Eq.\eqref{eq:Dxz+_first} can be expanded as
\begin{equation}\label{eq:Dxz+}
D^{+}_{xz}(D_{xz},\nu,\lambda_{xz})=D_{xz}+B_{D,xz}(\nu)\left[\lambda_{xz}B^\bot_{d_0,xz}(\nu)+\Delta V^0_{xz}(d_0,\nu)\right].
\end{equation}
Then, considering an unconstrained control, the in-plane instantaneous reachable set can be described as a one dimensional set
\begin{equation}\label{setDelta+xz}
\Delta^+_{xz}(D_{xz},\nu)=\{D^+_{xz}(D_{xz},\nu,\lambda_{xz})\in\R^4~~\text{s.t.}~~\text{Eq.}\eqref{eq:Dxz+},~~ \lambda_{xz}\in\R \}.
\end{equation}
Again, dead-zone and saturation constraints are formally described by the set $\Lambda_{\text{sat},xz}$ such that
\begin{equation}\label{eq:Isat_xz}
\begin{split}
\Lambda_{\text{sat},xz}(d_0,\nu)=\{\lambda_{xz} \in \R \quad \text{s.t.} \quad \underline{\Delta V} \leq \lVert \lambda_{xz}B_{d_{0,xz}}^\bot(\nu) + \Delta V^0_{xz}(d_0,\nu) \rVert_2 \leq \overline{\Delta V}\}=\\
[\overline{\lambda}_{xz,1}(d_0,\nu),~\underline{\lambda}_{xz,1}(d_0,\nu)] \cup [\underline{\lambda}_{xz,2}(d_0,\nu),~\overline{\lambda}_{xz,2}(d_0,\nu)],
\end{split}
\end{equation}
being $\underline{\lambda}_{xz,1}$, $\underline{\lambda}_{xz,2}$, $\overline{\lambda}_{xz,1}$ and $\overline{\lambda}_{xz,2}$
\begin{align}
\underline{\lambda}_{xz,1},~\underline{\lambda}_{xz,2}=-(B_{d_{0,xz}}^\bot)^T\Delta V^0_{xz}\pm\sqrt{\left((B_{d_{0,xz}}^\bot)^T\Delta V^0_{xz}\right)^2-\lVert\Delta V^0_{xz}\rVert^2_2+\underline{\Delta V}^2},\\
\overline{\lambda}_{xz,1},~\overline{\lambda}_{xz,2}=-(B_{d_{0,xz}}^\bot)^T\Delta V^0_{xz}\pm\sqrt{\left((B_{d_{0,xz}}^\bot)^T\Delta V^0_{xz}\right)^2-\lVert\Delta V^0_{xz}\rVert^2_2+\overline{\Delta V}^2}.
\end{align}
Note that the dependencies with $d_0$ and $\nu$ have been omitted for the sake of clarity. Therefore, the instantaneous in-plane reachable set, accounting for dead-zone and saturation, is given by
\begin{equation}\label{eq:Deltaxzsat}
\begin{aligned}
\Delta^+_{\text{sat},xz}(D_{xz},\nu)=&\left\lbrace D^+_{xz}(D_{xz},\nu,\lambda_{xz})\in\R^4~~\text{s.t.}~~D^+_{xz}=D_{xz}+B_{D,xz}(\nu)\left[\lambda_{xz}B^\bot_{d_0,xz}(\nu)+\Delta V^0_{xz}(d_0,\nu)\right]\right.,\\
&\lambda_{xz}\in \Lambda_{\text{sat},xz}(d_0,\nu)\Big\}.
\end{aligned}
\end{equation}

\subsubsection{Admissible set reachability conditions}

For the sake of clarity, the time dependence $\nu$ and state dependence $D$ are omitted in this section.

A necessary and sufficient condition for the admissible set $\adset{}$ to be reachable from the current state is that the following sets $\Delta^S_{\text{sat},xz}$ and $\Delta^S_{\text{sat},y}$ are non-empty:
\begin{align}
\Delta^S_{\text{sat},xz}&=\Delta^+_{\text{sat},xz}\cap \adsetxz,\\
\Delta^S_{\text{sat},y}&=\Delta^+_{\text{sat},y} \cap \adsety,
\end{align}
where $\adsetxz$ and $\adsety$ denote the $\adset$ projections in the $D_{xz}$ and $D_y$ spaces respectively. The sets $\Delta^S_{\text{sat},xz}$ and $\Delta^S_{\text{sat},y}$ are defined as the intersection of the sets $\Delta^+_{\text{sat},xz}$ and $\Delta^+_{\text{sat},y}$ (which are parameterized by a single parameter, $\lambda_{xz}$ or $\lambda_y$) with $\adset{}$, the semi-algebraic set from Eq.\eqref{eq:adset_envelope}, respectively. Consequently, if non-empty, the sets $\Delta^S_{\text{sat},xz}$ and $\Delta^S_{\text{sat},y}$ are line segments parameterized by $\lambda_{xz}$ and $\lambda_{y}$. To obtain a tractable characterization of the previously defined sets, the following sets are stated
\begin{eqnarray}
\Delta^S_{xz}&=& \Delta^+_{xz}\cap \adsetxz,\\
\Delta^S_{y}&=& \Delta^+_{y}\cap \adsety.
\end{eqnarray}
As $\Delta^+_{xz}$ and $\Delta^+_y$ are lines in their respective spaces (see Eq.\eqref{setDelta+xz} and Eq.\eqref{setDelta+y}) and the admissible set $\adset$ is described through its envelope (cf. Eq.\eqref{eq:adset_envelope}), the sets $\Delta^S_{xz}$ and $\Delta^S_{y}$ are given by
\begin{eqnarray}
\Delta^S_{xz}&=&\{D^+_{xz}(\lambda_{xz})\in \Delta^+_{xz}~ |~ g_{\xmin}(D^+(\lambda_{xz}))\leq 0,~g_{\xmax}(D^+(\lambda_{xz}))\leq 0,~g_{\zmin}(D^+(\lambda_{xz}))\leq 0,~g_{\zmax}(D^+(\lambda_{xz}))\leq 0\},\\
\Delta^S_{y}&=&\{D^+_{y}(\lambda_{y})\in \Delta^+_{y}~ |~ g_{\ymin}(D^+(\lambda_{y}))\leq 0,~ g_{\ymax}(D^+(\lambda_{y}))\leq 0 \}.
\end{eqnarray}
The property of $D^+$ belonging to $\adset$ boils down to formal conditions on the control variables $\lambda_{xz}$ and $\lambda_y$ such that $\lambda_{xz}\in \Lambda^S_{xz}=[\underline{l_{xz}}, \overline{l_{xz}}]$ and $\lambda_{y}\in \Lambda^S_y=[\underline{l_{y}}, \overline{l_{y}}]$.

\begin{rem}\label{rem_connectedness}
If non-empty, the connectedness of $\Lambda^S_{xz}$ and $\Lambda^S_y$ and consequently of the sets $\Delta^S_{xz}$ and $\Delta^S_y$ is ensured by the convexity of $\adset$.
\end{rem}

The intervals $\Lambda^S_{xz}$ and $\Lambda^S_y$ are computed such that 
\begin{align}
\forall \lambda_{xz} \in [\underline{l_{xz}}, \overline{l_{xz}}]:&\{ g_{\overline{x}}(\lambda_{xz})\le 0,~ g_{\underline{x}}(\lambda_{xz})\le 0,~ g_{\overline{z}}(\lambda_{xz})\le 0,~ g_{\underline{z}}(\lambda_{xz})\le 0\}, \label{eq:interval_lambda_xz},\\
\forall \lambda_{y} \in [\underline{l_{y}}, \overline{l_{y}}]:&\{g_{\overline{y}}(\lambda_y)\le 0,~ g_{\underline{y}}(\lambda_y)\le 0\}. \label{eq:interval_lambda_y}
\end{align}
The univariate polynomials $g_w(\cdot)$ are obtained by introducing Eq.\eqref{eq:Dxz+} and Eq.\eqref{eq:Dy+} into the polynomial expressions $g_{\overline{x}}(d_1, d_2, d_3)$, $g_{\underline{x}}(d_1, d_2, d_3)$, $g_{\overline{z}}(d_1, d_2)$, $g_{\underline{z}}(d_1, d_2)$, $g_{\overline{y}}(d_4, d_5)$ and $g_{\underline{y}}(d_4, d_5)$ respectively. Therefore, the intervals bounds $\underline{l_{xz}}$, $\overline{l_{xz}}$, $\underline{l_{y}}$ and $\overline{l_{y}}$ can be computed as roots of the univariate polynomials arising in Eq.\eqref{eq:interval_lambda_xz}-\eqref{eq:interval_lambda_y} respectively. 
The out-of-plane ($g_{\ymin}$, $g_{\ymax}$) and radial ($g_{\zmin}$, $g_{\zmax}$) constraints  are quadratic polynomials in $\lambda_y$ and $\lambda_{xz}$ respectively, whereas the in-track ($g_{\xmin}$, $g_{\xmax}$) constraints are sextic polynomials in $\lambda_{xz}$. 
Following the remark \ref{rem_connectedness}, the existence of two real roots is ensured if the sets $\Delta^S_{xz}$ and $\Delta^S_{y}$ are non-empty. On the contrary, the absence of real roots reveals that  $\Delta^S_{xz}$ and $\Delta^S_{y}$ are empty sets. Note that polynomial roots computations are efficiently executed with most of numerical scientific libraries.

Taking into account the dead-zone and saturation conditions of Eq.\eqref{eq:Isat_xz} and Eq.\eqref{eq:Isat_y}, $\adset$ is instantaneously reachable if and only if the sets $\Delta^S_{\text{sat},xz}$ and $\Delta^S_{\text{sat},y}$ are non-empty
\begin{align}
\Delta^S_{\text{sat},xz}&=\{D_{xz}^+(\lambda_{xz})| \lambda_{xz} \in \Lambda^S_{\text{sat},xz} \},\\
\Delta^S_{\text{sat},y}&=\{D_{y}^+(\lambda_{y}) | \lambda_{y} \in  \Lambda^S_{\text{sat},y} \},
\end{align}
where $\Lambda^S_{\text{sat},xz}= \Lambda^S_{xz}\cap \Lambda_{\text{sat},xz}$ and  $\Lambda^S_{\text{sat},y}= \Lambda^S_{y}\cap \Lambda_{\text{sat},y}$. Since $\Lambda_{\text{sat},xz}$ and $\Lambda_{\text{sat},y}$ are not connected sets, their intersections with $\Lambda^S_{xz}$ and $\Lambda^S_y$ do not yield connected sets either
\begin{align}
\Lambda^S_{\text{sat},xz}=&[l_{xz1},~l_{xz2}]\cup[l_{xz3},~l_{xz4}], \label{eq:lambda_sat_xz}\\
\Lambda^S_{\text{sat},y}=&[l_{y1},~l_{y2}]\cup[l_{y3},~l_{y4}]. \label{eq:lambda_sat_y}
\end{align}
To assess the $\adset$ proximity, it is useful to define the following variables measuring the total length of the intervals composing $\Lambda^S_{\text{sat},xz}$ and $\Lambda^S_{\text{sat},y}$ 
\begin{align}
L_{xz}=\mathrm{len}(\Lambda^S_{\text{sat},xz}),& \quad L_y=\mathrm{len}(\Lambda^S_{\text{sat},y}).
\end{align}
These indicators suggest that the admissible set is reachable if and only if both of them differ from 0
\begin{equation}
\Delta^S_{\text{sat}}=\Delta^S_{\text{sat},xz}\times\Delta^S_{\text{sat},y}\neq \emptyset \Leftrightarrow L_{xz}\neq 0 \land L_y \neq 0. 
\end{equation}
However, the lengths $L_{xz}$ and $L_y$ are not well-posed indicators to effectively assess $\adset$ proximity by detecting $\adset$ reachability opportunities in a continuous manner. This is due to the fact that $\adset$, though a convex set, is defined as the interior region resulting from the intersection of several semi-algebraic sets described by multivariate polynomials \eqref{eq:xmin_envelope_interior}-\eqref{eq:zmax_envelope}. 
As a consequence, $\adset$ may have edges and vertexes where the disjoint lines $\Delta^S_{\text{sat},xz}$ and $\Delta^S_{\text{sat},y}$ (intersections of the instantaneous reachable set with $\adset$) could instantaneously vanish, thus $L_{xz}$ and $L_y$ are not guaranteed to be continuous functions with time. To overcome this issue, alternate instantaneous reachability metrics (employed in paragraph \ref{single_impulse_control_law} as instantaneous reachability indicators) are defined as
\begin{equation}
\begin{aligned}
G_{xz}&=\begin{dcases}
\max\{g^*_{\xmin},g^*_{\xmax},g^*_{\zmin},g^*_{\zmax}\},~~&\text{if}~~L_{xz}>0,\\
0,~~&\text{if}~~L_{xz}=0,\end{dcases} &G_{y}&=\begin{dcases}\max\{g^*_{\ymin},g^*_{\ymax}\},~~&\text{if}~~L_{y}>0,\\
0,~~&\text{if}~~L_{y}=0,\end{dcases}\\
G_{\nu,xz}&=\frac{dG_{xz}}{d\nu},& G_{\nu,y}&=\frac{dG_{y}}{d\nu},
\end{aligned}\label{eq:epsilon_indicator}
\end{equation}
where 
\begin{equation}
\begin{aligned}
&g^*_{\xmin}=\min_{\lambda_{xz}}g_{\xmin}(\lambda_{xz})~~&\text{s.t.}&~~\lambda_{xz}\in\Lambda^S_{\text{sat},xz},&\quad&g^*_{\xmax}=\min_{\lambda_{xz}}g_{\xmax}(\lambda_{xz})~~&\text{s.t.}&~~\lambda_{xz}\in\Lambda^S_{\text{sat},xz},\\
&g^*_{\ymin}=\min_{\lambda_{y}}g_{\ymin}(\lambda_{y})~~&\text{s.t.}&~~\lambda_{y}\in\Lambda^S_{\text{sat},y},&\quad&g^*_{\ymax}=\min_{\lambda_{y}}g_{\ymax}(\lambda_{y})~~&\text{s.t.}&~~\lambda_{y}\in\Lambda^S_{\text{sat},y},\\
&g^*_{\zmin}=\min_{\lambda_{xz}}g_{\zmin}(\lambda_{xz})~~&\text{s.t.}&~~\lambda_{xz}\in\Lambda^S_{\text{sat},xz},&\quad&g^*_{\zmax}=\min_{\lambda_{xz}}g_{\zmax}(\lambda_{xz})~~&\text{s.t.}&~~\lambda_{xz}\in\Lambda^S_{\text{sat},xz},\\
\end{aligned}
\end{equation}
that is $g^*_{\xmin}$, $g^*_{\xmax}$, $g^*_{\ymin}$, $g^*_{\ymax}$, $g^*_{\zmin}$ and $g^*_{\zmax}$ are the $\adset$ univariate polynomials, in $\lambda_{xz}$ and $\lambda_y$, minimums for $\lambda_{xz}\in\Lambda^S_{\text{sat},xz}$ and $\lambda_{y}\in\Lambda^S_{\text{sat},y}$. These minimums are computed through a zeros search (roots computation) of the $\adset$ polynomials derivative with respect to $\lambda_{xz}$ or $\lambda_y$ respectively. The variables $G_{xz}$ and $G_y$ are guarantee to be continuous signals over time. Note that, due to the $\max$ operator in Eq.\eqref{eq:epsilon_indicator}, $G_{\nu,xz}$ and $G_{\nu,y}$ derivatives shall not be continuous in general.      

Figure \ref{fig:sketch_event_strategy} illustrates the reachability conditions (and associated metrics) for the in-plane motion at three different instants ($\nu_1$, $\nu_2$, $\nu_3$). Considering that the state is close enough to the admissible set, $D$ can be assumed near the equilibrium e.g. $|d_0|\approx 0$ (quasi-equilibrium assumption). Therefore, in this illustration, the state $D$ is assumed to be constant over time. It can be observed that $L_{xz}(D,\nu_1) > L_{xz}(D,\nu_2) > L_{xz}(D,\nu_3)=0$ (with $G_{xz}(D,\nu_1)<G_{xz}(D,\nu_2)<G_{xz}(D,\nu_3)=0$). This suggests that the single impulse must be applied before the length $L_{xz}$ vanishes (or equivalently $G_{xz}$ becomes null) and $\adset$ becomes unreachable with a single impulse. However, at $\nu=\nu_3$ where $\adset$ is instantaneously unreachable, reachability opportunities are predicted to come along the next orbit period as it was assumed $D$ is constant. Under this quasi-equilibrium assumption, both signals $L_{xz}$ and $G_{xz}$ are $2\pi$ periodic: $L_{xz}(\nu) \approx L_{xz}(\nu+2\pi)$ and $G_{xz}(\nu) \approx G_{xz}(\nu+2\pi)$. This is due to the 2-$\pi$ periodicity of the $D$ state control matrix $B_D(\nu) = B_D(\nu+2\pi)$, see Eq.\eqref{eq:B_D}, and the quasi-equilibrium assumption. 
Consequently, new control opportunities are expected at the instants $2\pi+\nu_1$ and $2\pi+\nu_2$ again. This statement also applies to the out-of-plane motion which is naturally periodic.    

\begin{figure}[h] 
	\begin{center}
		\includegraphics[width=12cm,height=12cm,keepaspectratio]{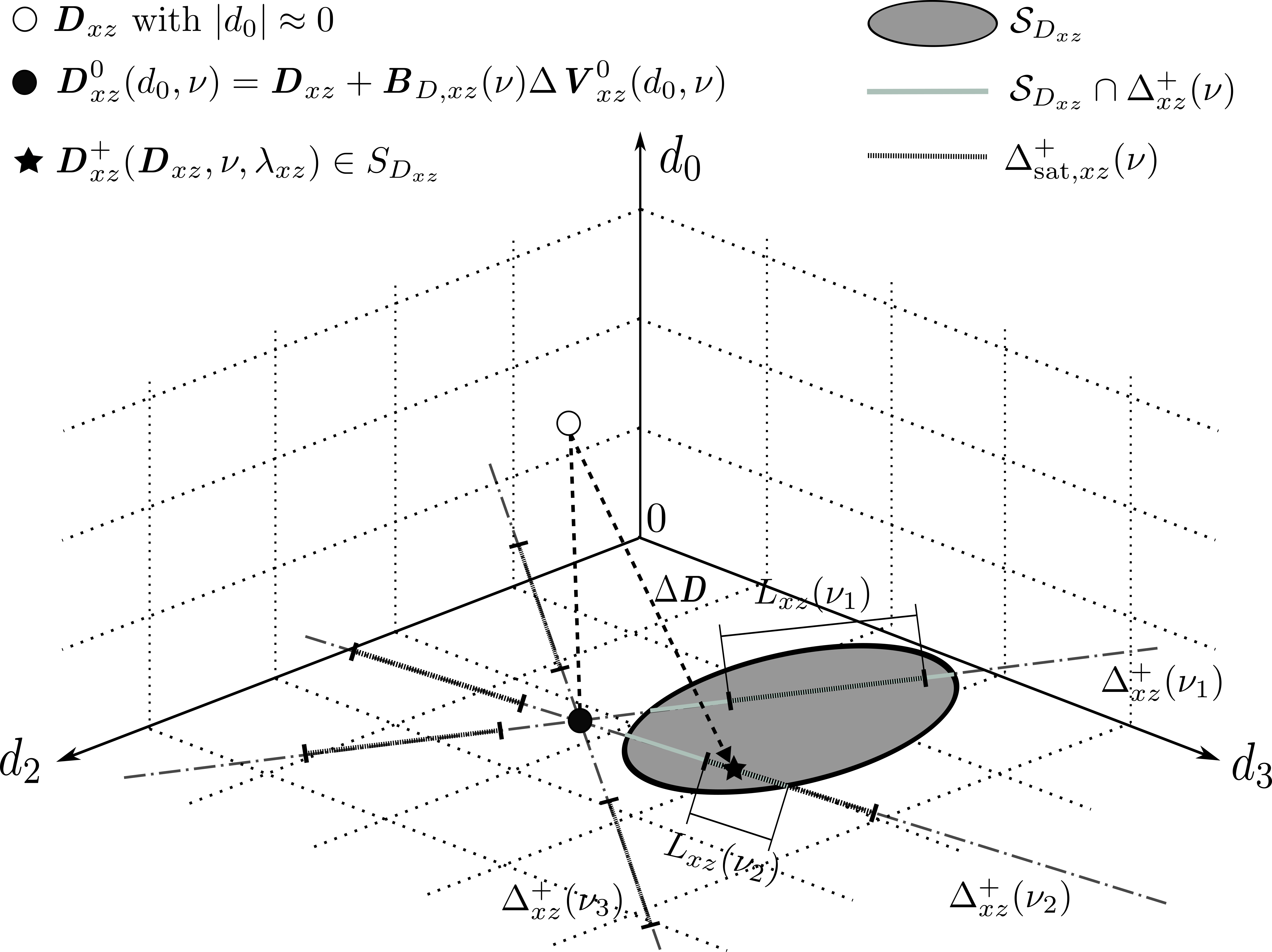}
	\end{center}
	\caption{Sketch illustrating the admissible set instantaneous reachability for $\nu_1\leq\nu_2\leq\nu_3$.}	
	\label{fig:sketch_event_strategy}
\end{figure}

\subsection{Region of attraction}\label{reachability_one_period}

In the previous section, instantaneous reachability conditions to $\adset$ and related metrics have been established for a given time $\nu$. Extending the previous conclusions, this section establishes the admissible set region of attraction $\mathcal{D}$. The region of attraction is defined as the set of states $D$ from which $\adset$ is instantaneously reachable for a short period of time taking place along the next 2-$\pi$-period:
\begin{equation}\label{eq:D_set}
\mathcal{D}=\mathcal{D}_{xz}\times\mathcal{D}_y,
\end{equation}
where
\begin{align}\label{eq:Dxz_set}
\mathcal{D}_{xz}:=&\left\{D\in\mathbb{R}^4: \>D\not\in\adsetxz, \> \int^{\nu+2\pi}_{\nu} L_{xz}(D_{xz}(\tau),\tau)d\tau \neq 0\right\},\\
\mathcal{D}_{y}:=&\left\{D\in\mathbb{R}^2: \>D\not\in\adsety, \> \int^{\nu+2\pi}_{\nu} L_{y}(D_y,\tau)d\tau \neq 0 \right\}.\label{eq:Dy_set}
\end{align}
Note that the in-plane state $D_{xz}$ varies if $d_0\neq0$, see Eq.\eqref{eq:D_trans_mat}. Such an attractive set $\mathcal{D}$ is of particular interest for the set of triggering rules to be defined in the sequel. As a matter of fact, if the current state belongs to $\mathcal{D}$, the impulsive control can be postponed for another reachability opportunity, over the next 2-$\pi$ period, even if $\adset$ is not instantaneously reachable at the moment. 

The definition of Eq.\eqref{eq:D_set} also provides the condition for the state $D$ to belong to the region of attraction $\mathcal D$:
\begin{equation}
	\int^{\nu+2\pi}_{\nu} L_{xz}(D_{xz}(\tau),\tau)d\tau \neq 0~~ \wedge~~ \int^{\nu+2\pi}_{\nu} L_{y}(D_y,\tau)d\tau \neq 0~~	\Longleftrightarrow ~~D\in \mathcal D.
\end{equation}
For the sake of computational efficiency, the previous integrals of Eq.\eqref{eq:Dxz_set}-\eqref{eq:Dy_set} are approximated through a discretization scheme:
\begin{equation}
	\begin{array}{cl}
	\displaystyle\frac{1}{2\pi}\int_\nu^{\nu+2\pi} L_{xz}(D_{xz}(\tau),\tau) d\tau &\displaystyle \approx \sum^{n_L}_{p=1} L_{xz}(D_{xz}(\nu_p),\nu_p),\\
	\displaystyle\frac{1}{2\pi}\int_\nu^{\nu+2\pi} L_y(D_y,\tau) d\tau &\displaystyle \approx \sum^{n_L}_{p=1} L_y(D_y,\nu_p),\\
	\end{array}
	 \> \> \nu_p=\nu+\tfrac{2\pi p}{n_L},~~n_L\in\N.\label{eq:reachability_one_period}
\end{equation}
Using the previous approximation, a sufficient condition for the state $D$ to belong to the region of attraction $\mathcal D$ can be stated:
\begin{equation}\label{SufCond_Din}
	 \nu_p=\nu+\tfrac{2\pi p}{n_L},~~n_L\in\N,\quad
	 \begin{cases}
	 	\displaystyle\sum^{n_L}_{p=1} L_{xz}(D_{xz}(\nu_p),\nu_p)\neq 0,\\
		\displaystyle\sum^{n_L}_{p=1} L_y(D_y,\nu_p)\neq 0,\\
	 \end{cases}
	\Longrightarrow D\in\mathcal D.
\end{equation}
The condition given by Eq.\eqref{SufCond_Din}, is the one employed in the up-coming event-based control algorithm.

\subsection{Single-impulse control computation}\label{single_impulse_control_law}

Subsequently, the single-impulse programs to be solved are presented. By using the previously defined sets $\Lambda^S_{\text{sat},xz}$ and $\Lambda^S_{\text{sat},y}$, in the control variables spaces, it can be shown that the single impulse control programs can be easily computed with only a few objective function evaluations.

\paragraph{Out-of-plane control}

To steer the state $D_y$ to the admissible set $\adsety$ at time $\nu$, the following program is considered
\begin{equation}\label{Psat_out}
\begin{array}
{c}
\displaystyle \min_{\lambda_{y}} \lVert \Delta V_y(\lambda_y) \rVert_1,\\
\text{s.t.}~
\begin{cases}
D^+_{y}(D_y,\nu,\lambda_y)\in \adsety,\\
\Delta V_y = \lambda_y,\\
\lambda_{y} \in \Lambda_{\text{sat},y}.\\
\end{cases}
\end{array}
\end{equation}
The constraints of \eqref{Psat_out} are equivalent to $\lambda_{y} \in \Lambda^S_{\text{sat},y}$, see Eq.\eqref{eq:lambda_sat_y}, hence
\begin{equation}\label{Psat_out_lambday}
\displaystyle \min_{\lambda_{y}} | \lambda_y |,~~\text{s.t.}~~\lambda_y \in \Lambda^S_{\text{sat},y},
\end{equation}
where the optimal candidates are the intervals composing $\Lambda^S_{\text{sat},y}$, given by Eq.\eqref{eq:lambda_sat_y}, extrema
\begin{equation}
\Lambda^*_{y}=\{\lambda_{y1},~\lambda_{y2},~\lambda_{y3},~\lambda_{y4}\}.
\end{equation}
The optimal solution can be easily found by evaluating four times the objective function and choosing the value that yields the minimum
\begin{equation}\label{Psat_out_reduced}
\lambda^*_y=\argmin_{\lambda_y \in \Lambda^*_y}(|\lambda_y|).
\end{equation}

\paragraph{In-plane control}

To steer the in-plane state $D_{xz}$ to the admissible set $\adsetxz$ at time $\nu$, the following program is considered
\begin{equation}\label{Psat_in}
\begin{array}
{c}
\displaystyle \min_{\lambda_{xz}} \lVert \Delta V_{xz}(\lambda_{xz}) \rVert_1,\\
\text{s.t.}~
\begin{cases}
D^+_{xz}(D_{xz},\nu,\lambda_{xz})\in \adsetxz,\\
\Delta V_{xz}(D_{xz},\nu)=\lambda_{xz}B_{d_{0},xz}^\bot(D_{xz},\nu)+\Delta V^0_{xz}(D_{xz}),\\
\lambda_{xz} \in \Lambda_{\text{sat},xz}(D_{xz},\nu).\\
\end{cases}
\end{array}
\end{equation}
As in the out-of-plane case, the constraints of \eqref{Psat_in} are equivalent to $\lambda_{xz} \in \Lambda^S_{\text{sat},xz}$, see Eq.\eqref{eq:lambda_sat_xz}. Expanding the objective function yields
\begin{equation}\label{Psat_in_lambdaxz}
\displaystyle \min_{\lambda_{xz}} (|\lambda_{xz}B_{d_{0,x}}^\bot+\Delta V^0_{x}|+|\lambda_{xz}B_{d_{0,z}}^\bot+\Delta V^0_{z}|),~~\text{s.t.}~~\lambda_{xz} \in \Lambda^S_{\text{sat},xz},
\end{equation}
where the optimal candidates are the intervals composing $\Lambda^S_{\text{sat},xz}$ extrema and the points where the objective function slope changes due to the presence of the absolute value function
\begin{equation}
\Lambda^*_{xz}=\{\lambda_{xz1},~\lambda_{xz2},~\lambda_{xz3},~\lambda_{xz4},~-\Delta V^0_{x}/B_{d_{0,x}}^\bot,~-\Delta V^0_{z}/B_{d_{0,z}}^\bot\},
\end{equation}
thus the optimal solution is given by
\begin{equation}\label{Psat_in_reduced}
\lambda^*_{xz}=\argmin_{\lambda_{xz}\in \Lambda^S_{\text{sat},xz}\cap\Lambda_{xz}^*}(|\lambda_{xz}B_{d_{0,x}}^\bot+\Delta V^0_{x}|+|\lambda_{xz}B_{d_{0,z}}^\bot+\Delta V^0_{z}|).
\end{equation}
Note that only six objective function evaluations are required to find the optimal solution.

\paragraph{Coupled motion control}  In previous cases, the control is assumed to be purely applied either as in-plane or out-of-plane. In the case where in-plane and out-of-plane motions require to be controlled at the same time, the following approach is proposed. Its aim is to account for the coupling of the in-plane and out-of-plane motions through the control constraints in Eq.\eqref{sat_dz_conds} while preserving the low complexity level of problems \eqref{Psat_out_lambday} and \eqref{Psat_in_lambdaxz}.\\
First recall that each control is systematically computed to steer the satellite on a periodic orbit, by satisfying Eq.\eqref{eq:d0+=0}, thus the in-plane control structure is given by Eq.\eqref{eq:deltaV_xz}. Then, the impulsive control $\Delta V\in \R^3$ is given by
\[\Delta V(\lambda_{xz},\lambda_y)=[B^\bot_{d_0,x}(D_{xz},\nu)\lambda_{xz}+\Delta V^0_x(D_{xz},\nu),~\lambda_y,~B^\bot_{d_0,z}(D_{xz},\nu)\lambda_{xz}+\Delta V^0_z(D_{xz},\nu)]^T.\]
Consequently, when the relative motion is coupled, the following program is considered to compute the control
\begin{equation}
\begin{aligned}
& \underset{\lambda_{xz},\lambda_y}{\text{min}}
& & \lVert\Delta V(\lambda_{xz},\lambda_y)\rVert_1, \\
& \text{s.t.} & & \Delta V=[B^\bot_{d_0,x}\lambda_{xz}+\Delta V^0_x,~\lambda_y,~B^\bot_{d_0,z}\lambda_{xz}+\Delta V^0_z]^T,\\
&&& \underline{\Delta V}\leq \rVert\Delta V\lVert_2 \leq \overline{\Delta V},\\
&&& \lambda_{xz},\lambda_y\in \Lambda^S_{\mathrm{sat}},\\
\end{aligned}
\end{equation}
where dependencies with time, $\nu$, and the in-plane state $D_{xz}$ have been omitted for the sake of clarity. Then, as the reachability conditions are met both for the in-plane and out-of-plane motion, the search set $\Lambda^S_{\mathrm{sat}}\in\R^2$ is non empty and is described by
\[\Lambda^S_{\mathrm{sat}} = \Lambda^S_{\mathrm{sat},xz} \times \Lambda^S_{\mathrm{sat},y},\]
which is a non connected set composed of four convex components.
These four components are boxes of $\R^2$ whose vertices are given the sixteen couples $(\lambda_{xz}, \lambda_y)$:
\[\left\{(\lambda_{xz}, \lambda_y)_{i=1,\dots,16}\right\}=\{(\lambda_{xz}, \lambda_y) \text{ s.t. } \lambda_{xz}, \lambda_y \in \{\lambda_{xz1},~\lambda_{xz2},~\lambda_{xz3},~\lambda_{xz4}\} \times \{\lambda_{y1},~\lambda_{y2},~\lambda_{y3},~\lambda_{y4}\}\}.\]
Then, the approach consists in assessing the control constraints of Eq.\eqref{sat_dz_conds} at $\left\{(\lambda_{xz}, \lambda_y)_{i=1,\dots,16}\right\}$, which are the vertices of the four $\Lambda^S_{\mathrm{sat}}$ components. 
Thanks to the convexity property of the four components of $\Lambda^S_{\mathrm{sat}}$, if any of these vertices meet Eq.\eqref{sat_dz_conds} conditions, the optimization problem has a non empty solution space. The optimal one is chosen among the admissible vertices. On the contrary, if none of these vertices is admissible, a policy prioritizing the in-plane motion control over the out-of-plane one is applied.

\subsection{Event-based control algorithm}\label{event_based_control_algorithm}

The event-based control algorithm is composed of the single impulse control law presented in Section \ref{single_impulse_control_law} associated to the trigger rules. These trigger rules are designed to achieve a threefold objective. Firstly, they have to ensure the in-plane \eqref{Psat_in_reduced} and out-of-plane \eqref{Psat_out_reduced} control programs feasibility when executed. Secondly, unnecessary commands must be avoided. Finally, Zeno phenomena should be precluded.

The proposed trigger rules are summarized by Algorithm \ref{alg1}. The computation and execution of a control impulse takes place when certain conditions are met. Each time the spacecraft leaves the admissible set $\adset{}$, based on Eq.\eqref{eq:adset_envelope}, the instantaneous reachability indicators $(G_{xz},~G_y,~G_{\nu,xz},~G_{\nu,y})$ are checked. If $G_{xz}$ or $G_y$ are equal or above predefined thresholds, $\delta_{xz}$ or $\delta_y$ respectively, and growing ($G_{\nu,xz}>0$ or $G_{\nu,y}>0$), a single impulse is commanded through programs \eqref{Psat_in_reduced} or \eqref{Psat_out_reduced}. Note that the thresholds should be negative $\delta_{xz}\leq0$ and $\delta_y\leq0$. Otherwise, if $G_{xz}=0$ or $G_y=0$, the reachability over one target orbit period is assessed e.g. the belonging of state $D$ to $\mathcal{D}$ is evaluated. If it is the case that $D\notin\mathcal{D}$ ($\adset$ is unreachable over the next 2-$\pi$ period) the controller of \cite{Arantes2019}, which globally stabilizes $\adset$ with three periodical impulses, is commanded. 

The choice of the thresholds $\delta_{xz}$ and $\delta_y$ will impact the controller behavior. If they are very negative, the single-impulse law is triggered more frequently. Otherwise, if the thresholds are close to zero, the single-impulse law will be triggered at a slow pace which could potentially increase the controller efficiency. However, there exists an upper bound for the triggering thresholds in order to ensure the single-impulse law is triggered. This upper bound could be computed by finding the highest possible values of $G_{xz}$ and $G_y$ along a target orbit period over the whole domain $\mathcal{D}$.  Choosing the thresholds below these upper bounds guarantees that the single impulse control will always be commanded for $D\in\mathcal{D}$ at some time between $\nu$ and $\nu+2\pi$
\begin{equation}
\begin{aligned}
\delta_{xz} &< \min\left\{
\sup_{D \in \mathcal {D}} \left( \max_{\nu\in[0,2\pi]} G_{xz}(D_{xz},\nu)\right)
\right\},\\
\delta_{y} &< \min\left\{
\sup_{D \in \mathcal {D}} \left( \max_{\nu\in[0,2\pi]} G_{y}(D_{y},\nu)\right)
\right\}.\label{eq:def-delta}
\end{aligned}
\end{equation}
\begin{algorithm}
	\caption{(Trigger rules)}
	\textbf{Input:} $D$, $\nu$\\
	\textbf{Output:} control decision
	\begin{algorithmic}[1]
	\If {$D\in \adset$}
	\State Wait.
	\ElsIf {$D\notin \adset$ and $D\in\mathcal{D}$}
		\If {$G_{xz}(D,\nu)\geq \delta_{xz}$ and $G_{\nu,xz}(D,\nu)>0$ and $G_{y}(D,\nu)\geq \delta_y$ and $G_{\nu,y}(D,\nu)>0$}
		\State Solve \eqref{Psat_in_reduced} and \eqref{Psat_out_reduced}, apply $\Delta V_{xz}$ and $\Delta V_{y}$.
		\ElsIf {$G_{xz}(D,\nu)\geq \delta_{xz}$ and $G_{\nu,xz}(D,\nu)>0$}
		\State Solve \eqref{Psat_in_reduced} and apply $\Delta V_{xz}$.
		\ElsIf {$G_{y}(D,\nu)\geq \delta_y$ and $G_{\nu,y}(D,\nu)>0$}
		\State Solve \eqref{Psat_out_reduced} and apply $\Delta V_{y}$.
		\Else {}
		\State Wait.
		\EndIf
	\ElsIf {$D\notin \adset$ and $D\notin\mathcal{D}$}
	\State Apply the global stabilizing controller of \cite{Arantes2019}.
	\EndIf
	\end{algorithmic}
	\label{alg1}
\end{algorithm}

\section{Invariance of the single impulse approach}\label{invariance_single_impulse_approach}

In this section, the invariance of the primal single impulse approach is assessed. Firstly, the well-posed behaviour of the system is studied by using some fundamental results of hybrid impulsive systems. Next, some invariance results for the single-impulse approach, even in the case of continuous disturbances, are demonstrated. These proofs require checking certain conditions for both the out-of-plane and in-plane motion.

\subsection{Well-posedness for hybrid impulsive systems}

The impulsively controlled relative motion between two space vehicles can be recasted as an hybrid impulsive system composed of the continuous flow dynamics given by Eq.\eqref{eq:D_dyn} and the instantaneous state changes of Eq.\eqref{eq:jumpD}. As a consequence, the main results of \cite{Haddad2006} regarding invariance principles for hybrid impulsive systems apply to the case under study. Consider the following hybrid impulsive system $\mathcal{G}$
\begin{equation}\tag{$\mathcal{G}$} \label{eq:hybrid_system}
\begin{array}{llll}
&D'(\nu)&=A_D(\nu) D, \> \> D(0)=D_0 \in \mathcal{D}, \> \> &(\nu,D(\nu)) \notin \mathcal{Z},\\
&\Delta D(\nu)&=B_D(\nu)\Delta V(\nu), \> \> & (\nu,D(\nu)) \in \mathcal{Z},
\end{array}
\end{equation}
where $\Delta D$ denotes the instantaneous change on the state $D$ due to an impulse. For the previously presented single impulse approach, the so-called jump set $\mathcal{Z}$, is given by the trigger rules defined in Algorithm \ref{alg1}. Note that the problem is time-dependent, albeit the system coefficients are 2-$\pi$ periodic. 

One can precisely write the set $\mathcal{Z}$ for the out-of-plane and in-plane dynamics (resp. $\mathcal{Z}_y$ and $\mathcal{Z}_{xy}$), so that $\mathcal{Z}=\mathcal{Z}_y \times \mathcal{Z}_{xz}$:
\begin{eqnarray}
\mathcal{Z}_y&=&\left\{(\nu,D(\nu)):D\notin S_{D_y},~D \in \mathcal{D}_y,~G_y(D,\nu)\geq \delta_y,~G_{\nu,y}(D,\nu)>0
\right\},\\
\mathcal{Z}_{xz}&=&\left\{(\nu,D(\nu)):D\notin S_{D_{xz}},~D \in \mathcal{D}_{xz},~G_{xz}(D,\nu)\geq\delta_{xz},~G_{\nu,xz}(D,\nu)>0
\right\}.
\end{eqnarray}

The initial condition $D_0$ is assumed to lie in the admissible set region of attraction $\mathcal{D}=\mathcal{D}_y\times \mathcal{D}_{xz}$.
If the initial condition is not in $\mathcal{D}$, the stabilizing controller of~\cite{Arantes2019} is applied, which is out of the scope of this work, thus the case is not considered.

Following~\cite[Chapter~2,~p.~13]{Haddad2006}, the assumptions guaranteeing the well-posedness of the state jump time instants are satisfied. The first assumption (A1) states that the trajectory can only enter the jump set through a point that belongs to its closure but not from the jump set itself. However, due to the jump set $\mathcal{Z}$ form and noting the definition of $\mathcal D$, points on its closure but not in the set can only possibly  be on the boundary of $\adset$, where the dynamics is stationary (since $d_0=0$). Therefore, they cannot leave $\adset$ and enter the jump set. The closure of $\mathcal Z$ is described for a state $D\in\mathcal D$, by $G_i(D,\nu)=\delta_i$ and $G_{\nu,i}(D,\nu)>0$, being $i$ the subscript for $xz$ or $y$. Moreover, by definition, $G_i(D,\nu)$ is a continuous function in terms of $\nu$ and $D$. Consequently, the only way for the state $D$ to reach the jump set is to go through its closure. The second assumption (A2) from \cite[Chapter~2,~p.~13]{Haddad2006}, requires ensuring that when a trajectory intersects the jump set it exits $\mathcal{Z}$ without returning to it (at least for some finite time). In the proposed approach, the jump set, $\mathcal Z$, is defined outside the admissible set so that $\mathcal Z \cap \adset = \emptyset$. When the state trajectory enters the jump set, the computed control impulse sends back the state to the admissible set $\adset$. Therefore, the assumption (A2) is satisfied by the control law construction. As both assumptions are  verified, the resetting times are ensured to be well-defined. Zeno phenomena is precluded as well because after one jump the state is placed outside $\mathcal Z$ in the $\adset$ where the dynamics is stationary. All these arguments guarantee that the solution to (\ref{eq:hybrid_system}) exists and is unique over a forward time interval.

Next, the invariance of the single impulse approach is analyzed with and without disturbances. It holds that the trigger law makes the admissible set invariant and attractive, even in the presence of disturbances, under certain conditions of the involved sets. These conditions are assessed in the sequel taking into account the problem parameters $\{\underline{x},\overline{x},\underline{y},\overline{y},\underline{z},\overline{z}\}$, $e$, $\underline{\Delta V}$ and $\overline{\Delta V}$. 

\subsection{Invariance under the single impulse approach}\label{sec-invariant1}

Next, a result guaranteeing the existence of an attractive invariant set for (\ref{eq:hybrid_system}) is shown. It is not possible to use \cite[Chapter~2,~p.~38]{Haddad2006}) due to the time-varying nature of the system. 

\begin{assumption}\label{Assump1}
Every state in the neighborhood of $\adset$ can reach $\adset$ using one  unconstrained impulse over a 2-$\pi$ period.
\end{assumption}

In other terms, for a given state $D$ in the vicinity of $\adset$, the set $(D+\mathcal F^{\infty})\cap \adset \neq \emptyset$ where $\mathcal F^{\infty}$ is the state increment unconstrained reachable set (see the appendix \ref{ReachableSet}). For the out-of-plane motion $\mathcal F_y^{\infty}$ is the whole space $\R^2$. For the in-plane motion $\mathcal F_{xz}^{\infty}$ is a conic surface (see appendix \ref{ReachableSet} for definitions and descriptions).

\begin{remark}
If assumption \ref{Assump1} is not satisfied for some states in the closed vicinity of $\adset$, those states do not belong to $\mathcal D$ by definition. Consequently, there exist paths that escape $\adset$ without any opportunity to be steered back to $\adset$ with a single impulse. In such condition, invariance can not be guaranteed.
\end{remark}

In order to provide invariance results, the dead-zone set, $\mathcal D_{\text{dz}}$ , needs to be declared. This set is defined as the set of states from where all the $\adset$ reachability opportunities over a 2-$\pi$ period fall below the dead-zone threshold. This set may exist or not depending on the conditions given in appendix \ref{DeadzoneSet}.

\begin{theorem}\label{th-invariance} 
Consider the impulsive dynamical system given by \ref{eq:hybrid_system}.
Define the set $\mathcal M=\mathcal D \cup \adset$. If the dead-zone set is empty, $\mathcal D_{\textup{dz}}=\emptyset$, then for $D(0)\in\mathcal M$, it holds that $D(\nu)\rightarrow\mathcal M$ as $\nu\rightarrow\infty$. 
\end{theorem}

\begin{proof}
Consider the assumption \ref{Assump1}. Therefore, for any state in the $\adset$ vicinity, there exist an unconstrained control steering the state back to $\adset{}$ over the next 2-$\pi$ period: $(D+\mathcal F^{\infty})\cap\adset\neq\emptyset$. For states very close to $\adset$, the control is small enough so that $(D+\mathcal{F}_{\text{dz}})\cap\adset\neq\emptyset$.
At this point, two cases are possible:
the set $(D+\mathcal F_{\text{sat}})\cap\adset$ may be empty or not.
Recalling that the dead-zone set is described by
\begin{equation}
	\mathcal D_{\text{dz}} = \{D\in\R^6 : (D+\mathcal{F}_{\text{dz}})\cap \adset \neq \emptyset~\wedge~(D+\mathcal F_{\text{sat}})\cap\adset = \emptyset\},
\end{equation}
then, if the dead-zone set is empty, any state in the closed neighborhood of $\adset{}$ verifies $(D+\mathcal F_{\text{sat}})\cap\adset\neq\emptyset$. In other terms, the closed neighborhood of \adset{} belongs to $\mathcal D$:
\begin{equation}
	\partial \mathcal D \cap \partial \adset = \partial \adset.
\end{equation}
Consequently, any $\adset$ escape trajectory is guaranteed to enter $\mathcal D$. Noting that the set $\mathcal D$ is contractive to $\adset{}$ terminates the proof.
\end{proof}

Theorem \ref{th-invariance}  indicates that, under the event-based control laws, the admissible set is attractive and invariant and the union set $\mathcal M$ is invariant.

\begin{remark}
If the dead-zone set is non empty, some states at the boundary of $\adset{}$ belong to $\mathcal {D}_{\textup{dz}}$ so that $\partial \mathcal D \cap \partial \adset \neq \partial \adset$ and $\partial \mathcal D_{\textup{dz}} \cap \partial \adset \neq \emptyset$. Then, there exist $\adset$ escape trajectories without entering the region of attraction $\mathcal D$. Therefore, the invariance property can not be generally ensured by the local controller and the global control may be triggered.
\end{remark}

\paragraph{Invariance under disturbances}

Next, the hybrid  system $\mathcal{G}$ is modified to account for continuous dynamical disturbances
\begin{equation}\tag{$\mathcal{G}'$} \label{eq:hybrid_system2}
\begin{array}{llll}
&D'(\nu)&=A_D(\nu) D+w(\nu,D(\nu)), \> \> D(0)=D_0 \in \mathcal{D} \> \> &(\nu,D(\nu)) \notin \mathcal{Z},\\
&\Delta D(\nu)&=B_D(\nu)\Delta V(\nu), \> \> & (\nu,D(\nu)) \in \mathcal{Z},
\end{array}
\end{equation}
where $w(\nu,D(\nu))$ is an unknown Lipschitz continuous disturbing function, which is assumed to behave in a way such that it does not modify the assumptions validity guaranteeing the well-posedness of the jump times and the absence of Zeno behaviour for system~\ref{eq:hybrid_system}. Firstly, the invariance can only be ensured if theorem \ref{th-invariance} conditions hold. Then, $\adset{}$ is attractive and the set $\mathcal M$  is an invariant set under the following condition: the continuous disturbances function $w(\cdot)$ has  to be bounded such that $\varphi_D(\nu,D_0,w)\in \mathcal M$ for $D_0\in\adset$ and $\nu\in [\nu_0,~\nu_0+2\pi]$, being $\varphi_D$ the bundle of flow trajectories under the disturbances. Such a conjecture is only a necessary to ensure invariance under disturbances. As a matter of fact, remaining in the $\adset$ region of attraction, $\mathcal D$, ensures that an opportunity will raise during the next orbital period as the jump set $\mathcal Z$ lives in $\mathcal D$. However, the case where the disturbances causes the state to drift away $\mathcal{D}$, before such opportunity comes up, is also probable.

\section{Numerical experiments}
\label{Results}

In this section, the simulation results using the proposed event-based controller are presented. The section is divided in three parts to analyze the effect of the main parameters ($e$, $\underline{\Delta V}$, $\overline{\Delta V}$) affecting the controller performances. 
The simulations are run in MATLAB, with an i7-8700 3.2 GHz CPU, using the Simulink model \cite{SimMatlab} which is based on the non-linear Gauss variational equations accounting for Earth $J_2$ oblateness effects and atmospheric drag as disturbances sources $w$.

The scenario parameters common to all the simulations are the following: the Earth gravitational constant ($\mu=398600.4~\text{km}^3/\text{s}^2$), the leader orbit elements ($h_p=605~\text{km}$, $i=98^{\circ}$, $\Omega=0^{\circ}$, $\omega=0^{\circ}$), the leader initial true anomaly ($\nu_0=0^{\circ}$), the follower initial relative state on the LVLH frame ($X=[300~\text{m},~400~\text{m},~-40~\text{m},~0~\text{m/s},~0~\text{m/s},~0~\text{m/s}]^T$) and the hovering box bounds ($\xmin=50~\text{m}$, $\xmax=150~\text{m}$, $\ymin=-25~\text{m}$, $\ymax=25~\text{m}$, $\zmin=-25~\text{m}$, $\zmax=25~\text{m}$). Regarding atmospheric drag, the target and chaser ballistic coefficients are assumed to be similar to the International Space Station, $B_t=139.80~\text{kg/m}^2$, and Automated Transfer Vehicle, $B_c=175.90~\text{kg/m}^2$.

The vehicle starts outside $\adset$ and reaches the admissible set during an approach phase (carried out by the global stabilizing controller). Once the vehicle enters $\adset$, the hovering phase (which is the scope of this work) begins with the activation of the event-based controller. 

Regarding the event-based controller parameters, the trigger rules are evaluated at a sampling rate of ($\Delta \nu=1^{\circ}$), the $\adset$ instantaneous reachability thresholds are chosen as ($\delta_{xz}=-3$, $\delta_y=-100$) and the discretization parameter to evaluate Eq.\eqref{SufCond_Din} is taken as ($n_L=100$). The hovering phase lasts for 10 target orbit periods ($\nu_f=\nu_0'+10\cdot2\pi$). Note that $\nu_0'$ is the hovering phase initial instant. On the other hand, the global stabilizing controller parameters are the number of impulses ($N=3$), true anomaly interval between impulses ($\tau_I=30^{\circ}$), true anomaly interval between sequence of impulses ($\tau_S=5^{\circ}$) and true anomaly interval to achieve periodicity ($\tau_E=5^{\circ}$), see \cite{Arantes2019}. The studied parameters are the eccentricity $e$, the dead-zone $\underline{\Delta V}$ and the saturation $\overline{\Delta V}$ in paragraphs \ref{e_results}, \ref{DZ_results} and \ref{Sat_results} respectively. The specific parameters values are detailed in each paragraph.  

\subsection{Impact of the eccentricity}\label{e_results}

Due to its important role in the relative dynamics, see Eq.\eqref{eq:D_dyn}, the control matrix $B_D$ given by Eq.\eqref{eq:B_D}, and the admissible set description, see Eq.\eqref{eq:xmin_envelope_interior}-\eqref{eq:zmax_envelope}, the eccentricity impact is  assessed in this section. To this end, 50 simulations for different eccentricities equispaced between 0 and 0.6 have been carried out. The dead-zone and saturation values are chosen as $\underline{\Delta V}=10^{-3}~\text{m/s}$ and $\overline{\Delta V}=0.1~\text{m/s}$.

Firstly, the event-based controller behavior is studied and then compared with the global controller proposed in \cite{Arantes2019}. Simulation results for $e=0$ and $e=0.6$ are presented in Fig.\ref{fig:3D_trajectory}. At each case, only the hovering phase trajectory is represented and analysed. Let recall, that the approach phase is ensured using a proper controller as the globally stabilizing one from \cite{Arantes2019}. In Fig.\ref{subfig:3D_trajectory_e_00} the polytopic constraints are respected while in Fig.\ref{subfig:3D_trajectory_e_06} a violation arises near the upper left corner though the trajectory naturally returns to the hovering region. However, the global controller is not triggered. This translates to the fact that even if the admissible set is not reachable (the spacecraft is outside of the hovering zone) it will be in the next period ($D\in\mathcal{D}$). The event-based controller behavior can be seen in \figref{fig:intersections_impulses}. This figure shows the instantaneous reachability signals $G_{xz}$ and $G_y$ along with the triggered impulses $\Delta V_{xz}$ and $\Delta V_y$. The instantaneous reachability signals $G_{xz}$ and $G_y$ evolve quasi-periodically. This confirms the assumption that the state is in the equilibrium set (or very close to it) and the range of control evolves periodically as suggested by the control matrix $B_D$, see Eq.\eqref{eq:B_D}.

Figure \ref{fig:e_calls} counts the number of calls to the single impulse controller and to the global controller made by  the event-based algorithm during the hovering phase. Consequently, the number of triggering events are counted. This number is between 6 and 19 with an average of 10.3 events. One can note that the global controller is never called during all the simulated cases. In other terms, the use of the local event-based controller (trough the primal single impulse approach) permits to not lose track of the admissible set. This illustrates the robustness and invariance of the set $\mathcal{D}$ even in the presence of disturbances. 

\begin{figure}[] 
	\begin{center}
		\subfigure[]{\includegraphics[width=12cm,height=12cm,keepaspectratio]{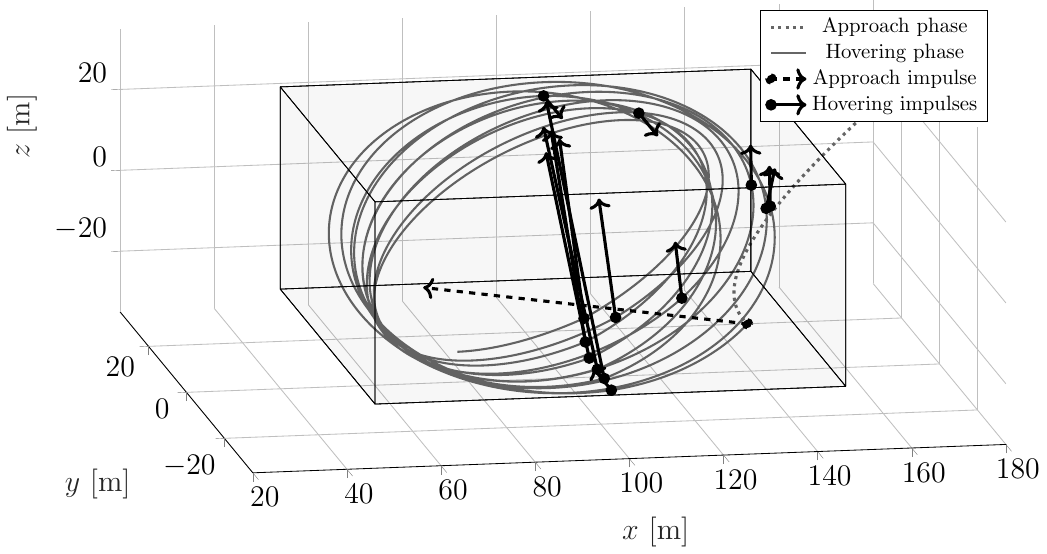} \label{subfig:3D_trajectory_e_00}}
		\subfigure[]{\includegraphics[width=12cm,height=12cm,keepaspectratio]{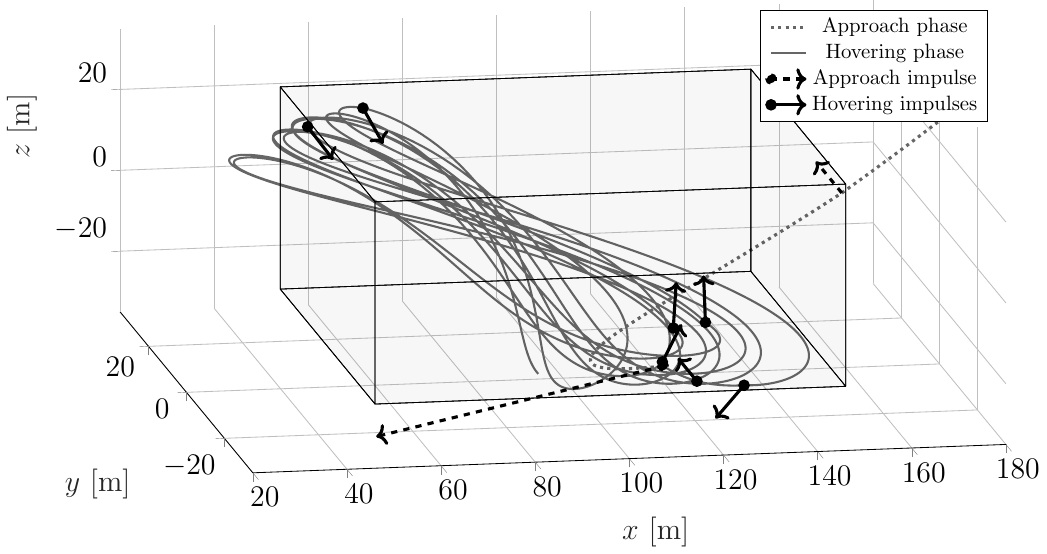} \label{subfig:3D_trajectory_e_06}}
		\caption{Trajectory for (a): $e=0$, (b): $e=0.6$. Hovering impulses scale 10:1 with respect to approach impulses.}
		\label{fig:3D_trajectory}
	\end{center}
\end{figure}

\begin{figure}[] 
	\begin{center}
		\subfigure[]{\includegraphics[width=13cm,height=13cm,keepaspectratio]{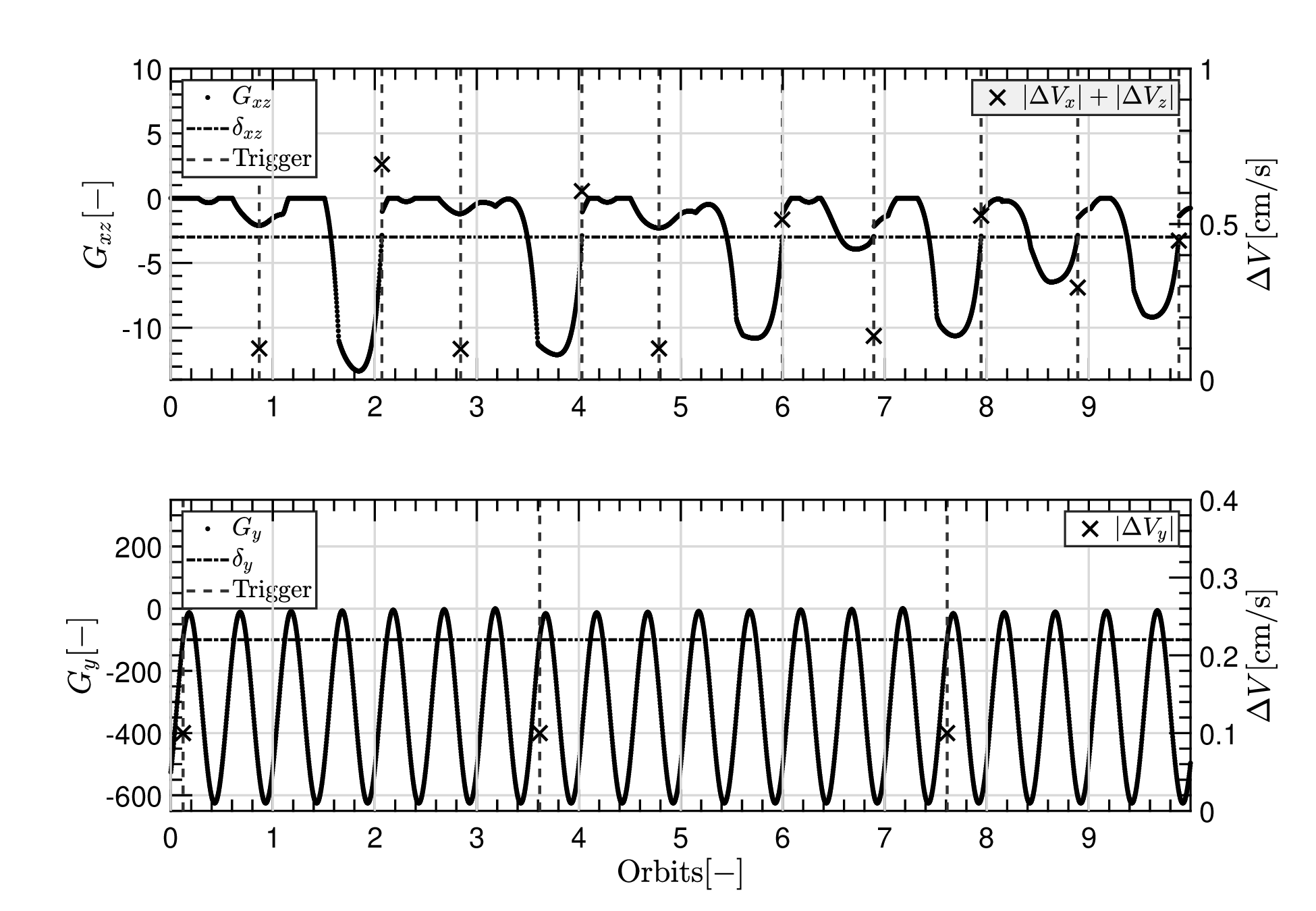} \label{subfig:intersections_impulses_e_00}}
		\subfigure[]{\includegraphics[width=13cm,height=13cm,keepaspectratio]{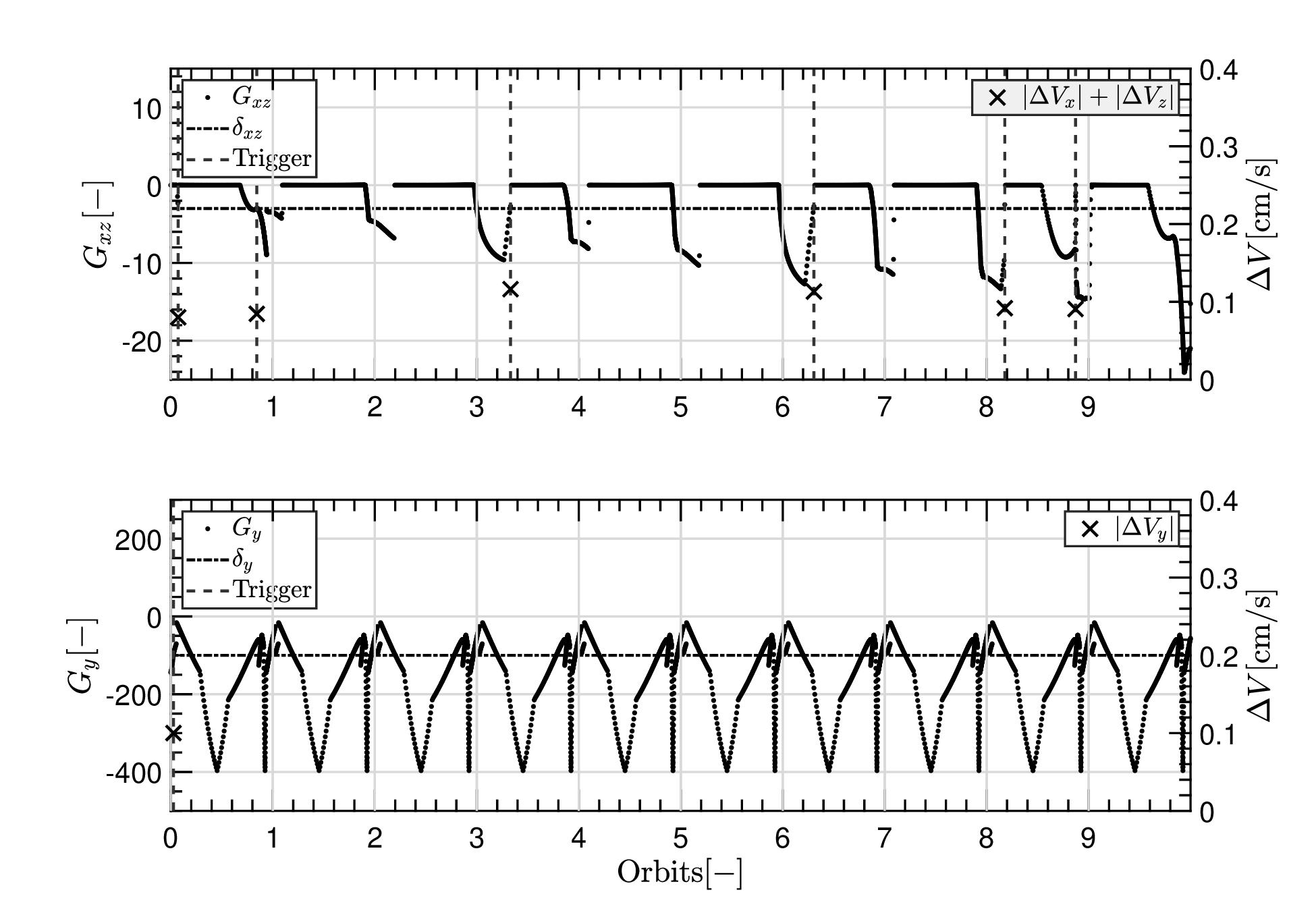}\label{subfig:intersections_impulses_e_06}}\qquad
		\caption{The variables $G_{xz}$, $G_y$, $\lVert \Delta V_{xz}\rVert_1$ and $\lVert \Delta V_y\rVert_1$ for (a): $e=0$, (b): $e=0.6$.}
		\label{fig:intersections_impulses}
	\end{center}
\end{figure}

\begin{figure}[] 
	\begin{center}
		\includegraphics[width=9.5cm,height=9.5cm,keepaspectratio]{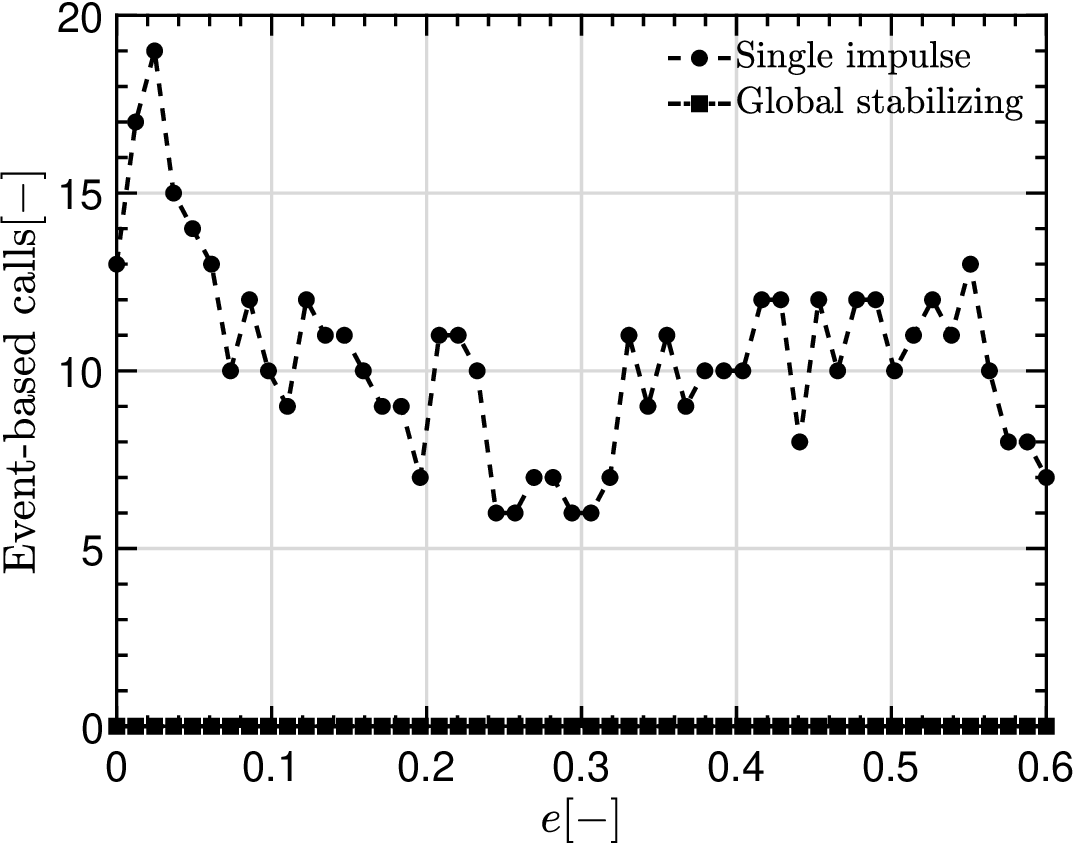}
	\end{center}
	\caption{Number of calls to each control law for the event-based controller.}	
	\label{fig:e_calls}
\end{figure}

Then, the event-based controller (by even studying very reactive thresholds $\delta_{xz}=-300$, $\delta_y=-10^4$) performance is compared with the periodic global stabilizing controller (considering two different tuning parameter values $\tau_I=30^{\circ},5^{\circ}$) proposed in \cite{Arantes2019}. The controllers accuracy is presented in Fig.\ref{fig:e_poly}. During the hovering phase, the hovering condition (the spacecraft position respects the polytopic constraints) is tested. Let recall that the relative state $D$ belonging to the admissible set, $\adset$, implies that the relative state position is within the hovering zone. Figure \ref{fig:e_poly} shows the percentage of time the chaser spacecraft remains within the hovering zone (cuboid). The event-based controller shows its superiority by ensuring an accuracy over 96.00\% (97.91\% for reactive thresholds) of the hovering time in almost all cases with an average value of 98.66\% (99.62\% for reactive thresholds). For the other controllers, the hovering mean time results are 95.43\% and 91.38\% (for $\tau_i=5^\circ$ and $\tau_i=30^\circ$ respectively). These values may be unacceptable, since each one-percent of constraints violation corresponds to an excursion time between $10~\text{minutes}$ (for $e=0$) and $40~\text{minutes}$ (for $e=0.6$). Note that the target perigee altitude is fixed, thus the semi-major axis and orbit period increase with the eccentricity. An explanation for this accuracy loss is that the global controller does not account for the thrusters minimum impulse bit. Consequently, most of the computed controls are not executed due to the dead-zone filtering (impulses below the dead-zone value are nullified) as it can be observed in Fig.\ref{fig:e_N}. As a matter of fact, the global controllers compute a total number of 723 and 203 impulsive controls during the hovering phase (depending on the parameter $\tau_i$ but not on the eccentricity value). However, only 3 to 21 of them are relevant with respect to the minimum impulse bit. On the contrary, the event-based controller computes between 6 and 19 controls for the nominal case while the reactive controller triggers 2 impulses more in average. It is worth noting that the number of relevant impulses is of the same magnitude for all the controllers while the accuracy clearly improves with the event-based approach.

The hovering phase fuel consumption, measured as $J=\sum^N_{i=1}||\Delta V||_1$, is shown in Fig.\ref{fig:e_J}. The cost is shown for both the event-based controller and global controllers (before and after dead-zone filtering). Before the dead-zone filtering (w/o filter), only the computed consumption is shown. The actual fuel consumption corresponds to the cases after dead-zone filtering. The event-based controller consumption is almost equivalent to the global controller actual consumption. More precisely, the event-based control consumes, at most, less than $4.5~\text{cm/s}$ for $0\leq e\leq0.1$, and less than $2~\text{cm/s}$ for $e>0.1$ typically. Similar results are yielded by the reactive thresholds which assure less than $3~\text{cm/s}$ for $0\leq e\leq0.1$ and the same performance for $e>0.1$ in average.

One of main arguments for the development of an event-based control algorithm is the enhancement in terms of numerical efficiency. The event-based controller computational load is composed of the trigger rules evaluation (at the specified sampling rate) and the in-plane or out-of-plane controls computation (when required). The most computationally consuming task is the trigger rules evaluation lasting an average of 5.895 miliseconds (see Table \ref{table_cputime}), whilst the computation times of the in-plane and out-of-plane controls are negligible in comparison (with average values of $0.0299~\text{ms}$ and $0.1314~\text{ms}$ respectively). One can note that the worst trigger rules computational time highly differs from the mean. This can be explained by the fact that in the case where the reachability over one period has to be computed through Eq.\eqref{eq:reachability_one_period} the trigger rules evaluation time increases significantly.

\begin{table}[h]
	\centering
	\begin{tabular}{lcccc}
		\hline \hline
		\multicolumn{1}{r}{[ms]} & Mean & Standard deviation & min & max\\
		\hline
		Trigger rules & 5.895 & 15.91 & 0.622 & 118.6\\ 
		In-plane control & 0.0299 & 0.0624 & 0.0270 & 0.5813  \\ 
		Out-of-plane control & 0.1314 & 0.0299 & 0.0209 & 0.4824\\ 
		\hline \hline
	\end{tabular}
	\caption{Event-based computation times.}
	\label{table_cputime}
\end{table}

Table \ref{table:e_tcpu} shows the cumulated computation times during the hovering phase (10 periods) for the different controllers. The cumulated computation times for the event-based controller are between 6 and 79 seconds, whilst, the global predictive controller spends between 223 and 266 seconds (for $\tau_i=5^{\circ}$), and between 60 and 72 seconds (for $\tau_i=30^{\circ}$). The superiority of the event-based controller with respect to the global controllers is justified by the fact that at five times ($\Delta\nu=1^{\circ}$) the global controller sampling rate ($\tau_I=5^{\circ}$), the event-based controller computational time is 20 times lower (in average). By selecting $\tau_i=30^\circ$, the global controller computational time is 3 times higher than the event-based one while the control accuracy is significantly lower with no gain in fuel consumption to mitigate these observations.

\begin{table}[h]
	\centering
	\begin{tabular}{lcccc}
		\hline \hline
		\multicolumn{1}{r}{[s]} & Mean & Standard deviation & min & max\\ \midrule
		Event-based controller & 21.223 & 13.831 & 5.811 & 79.128\\ 
		Global controller $\tau_i=5^\circ$ & 245.60 & 11.231 & 222.98 & 265.84\\ 
		Global controller $\tau_i=30^\circ$ & 66.500 & 3.143 & 60.427 & 72.191 \\
		\hline \hline
	\end{tabular}
	\caption{Cumulated computation times for the hovering phase.}
	\label{table:e_tcpu}
\end{table}

\begin{figure}[] 
	\begin{center}
		\includegraphics[width=10cm,height=10cm,keepaspectratio]{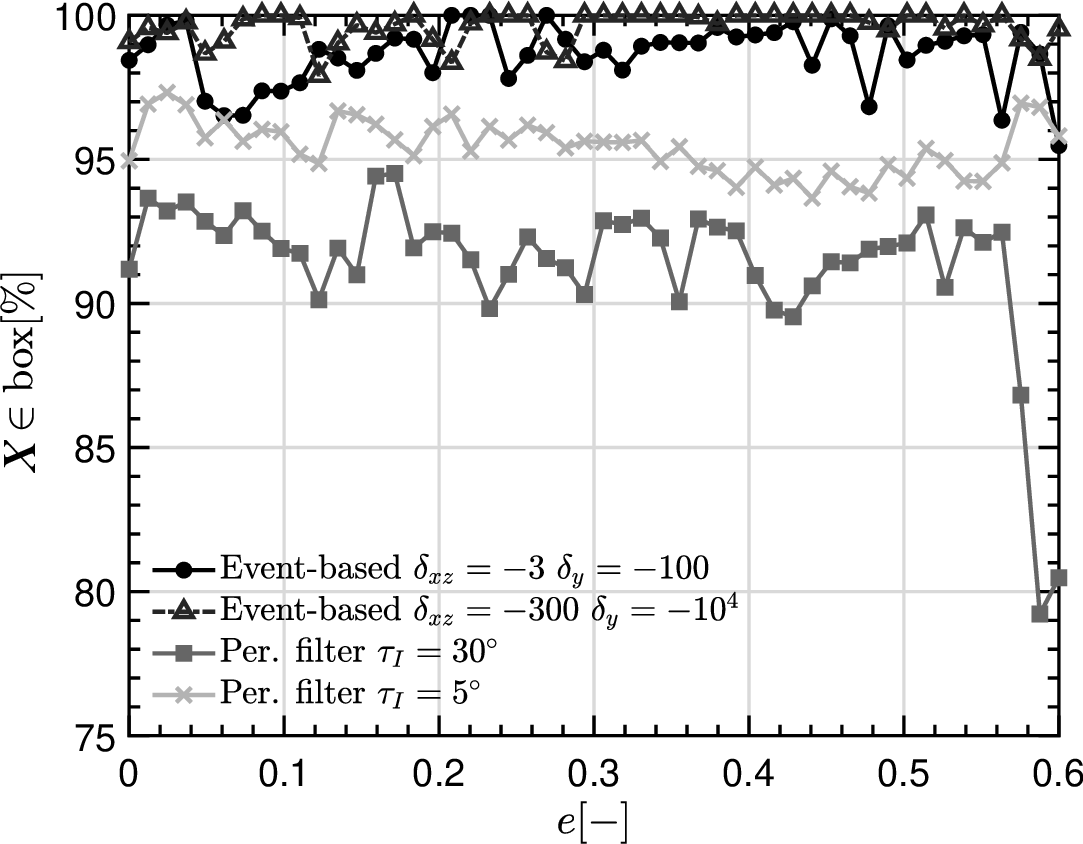}
	\end{center}
	\caption{Polytopic constraints satisfaction time percentage of the event-based and periodic controllers for different eccentricities.}	
	\label{fig:e_poly}
\end{figure}

\begin{figure}[] 
	\begin{center}
		\includegraphics[width=10cm,height=10cm,keepaspectratio]{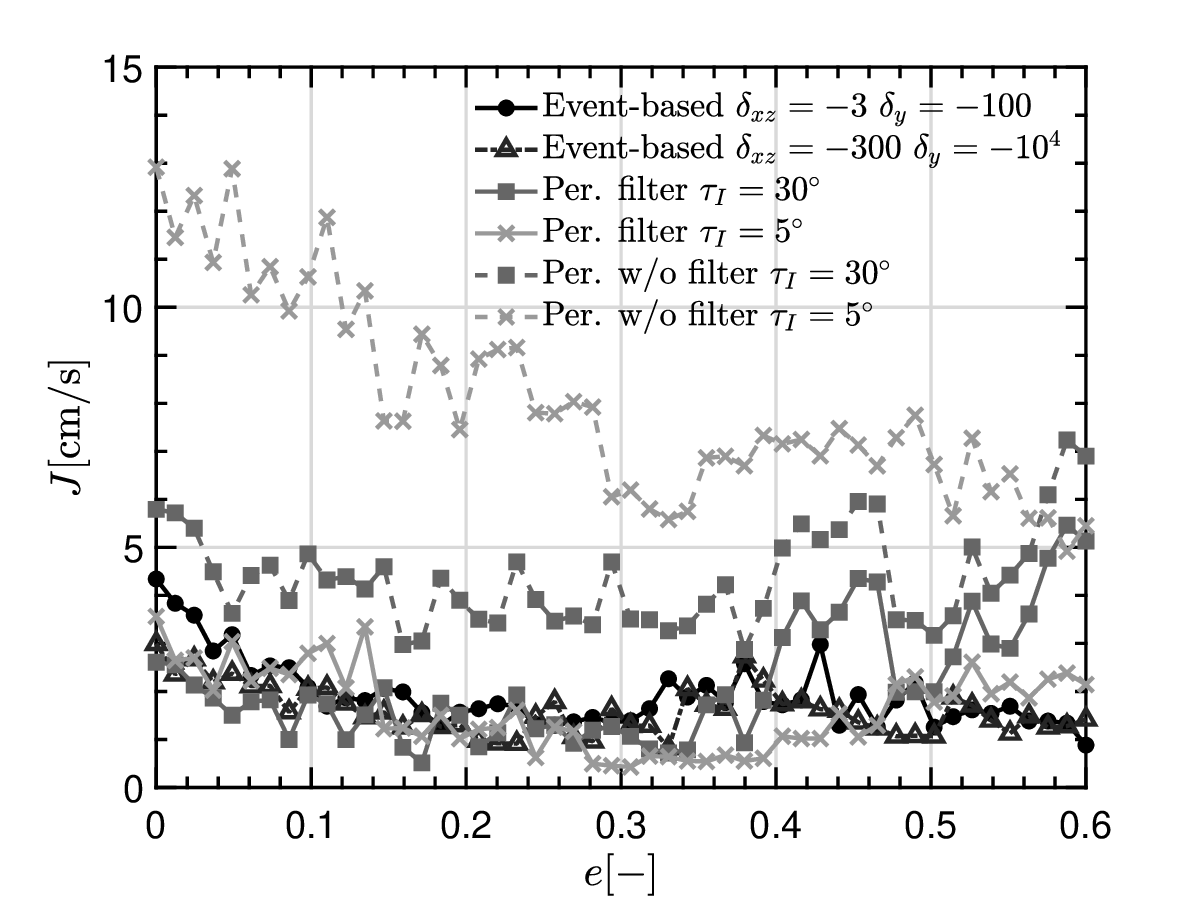}
	\end{center}
	\caption{Cost of the event-based and periodic controllers for different eccentricities. Filter$\equiv$impulses below dead-zone nullified; w/o filter$\equiv$impulses below dead-zone not nullified.}	
	\label{fig:e_J}
\end{figure}

\begin{figure}[] 
	\begin{center}
		\includegraphics[width=10cm,height=10cm,keepaspectratio]{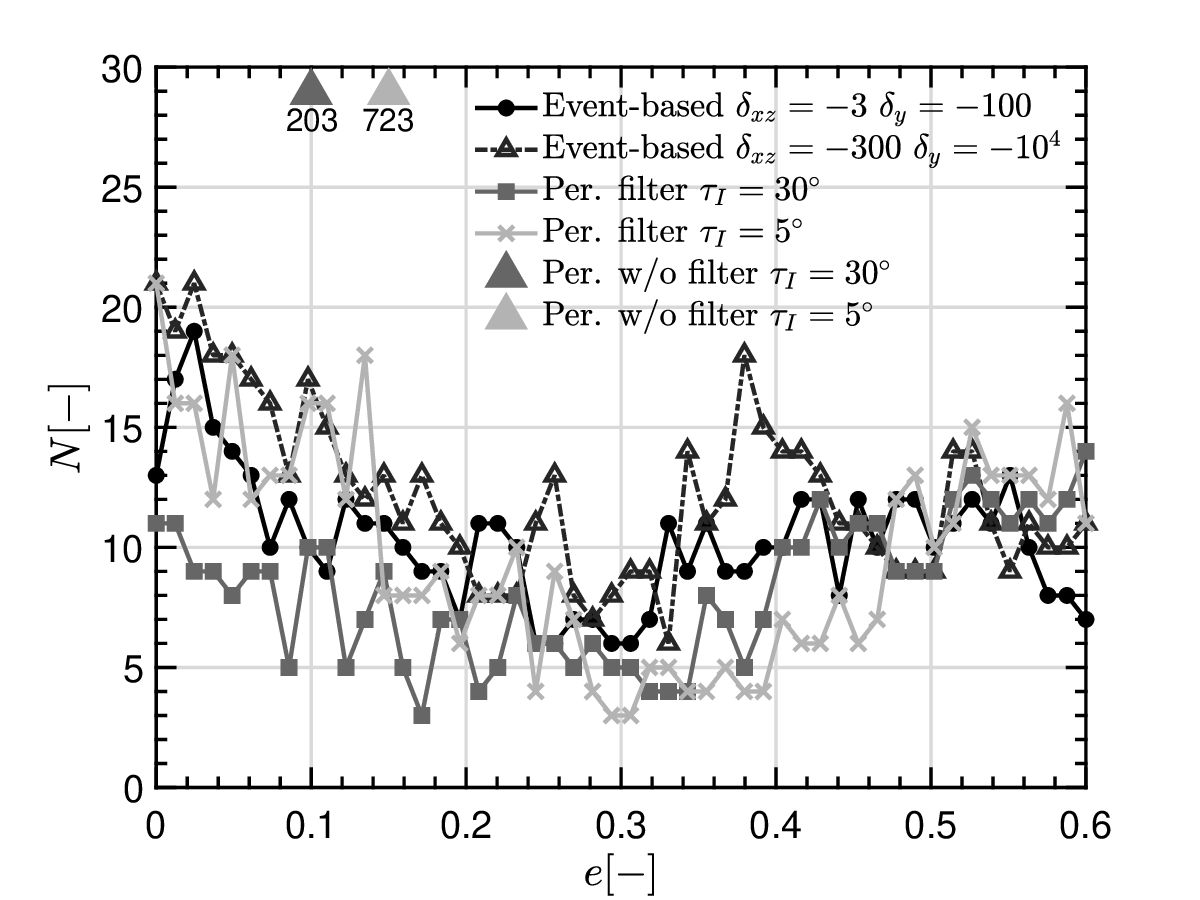}
	\end{center}
	\caption{Number of impulses of the event-based and periodic controllers for different eccentricities. The upper part of the figure shows the computed number of impulses for the periodic controllers (independent of eccentricity). Filter$\equiv$impulses below dead-zone nullified; w/o filter$\equiv$impulses below dead-zone not nullified.}	
	\label{fig:e_N}
\end{figure}

\subsection{Impact of the dead-zone}\label{DZ_results}

As shown in Section \ref{invariance_single_impulse_approach}, the dead-zone can significantly impact the event-based controller behavior. This is the reason why a parametric analysis on $\underline{\Delta V}$ is carried out for $e=0.004$ and $\overline{\Delta V}=0.1~\text{m/s}$. The dead-zone impact is assessed by simulating 50 values of $\underline{\Delta V}$, logarithmically equispaced between $10^{-4}~\text{m/s}$ and $10^{-2}~\text{m/s}$.  

\begin{figure}[] 
	\begin{center}
		\includegraphics[width=9.5cm,height=9.5cm,keepaspectratio]{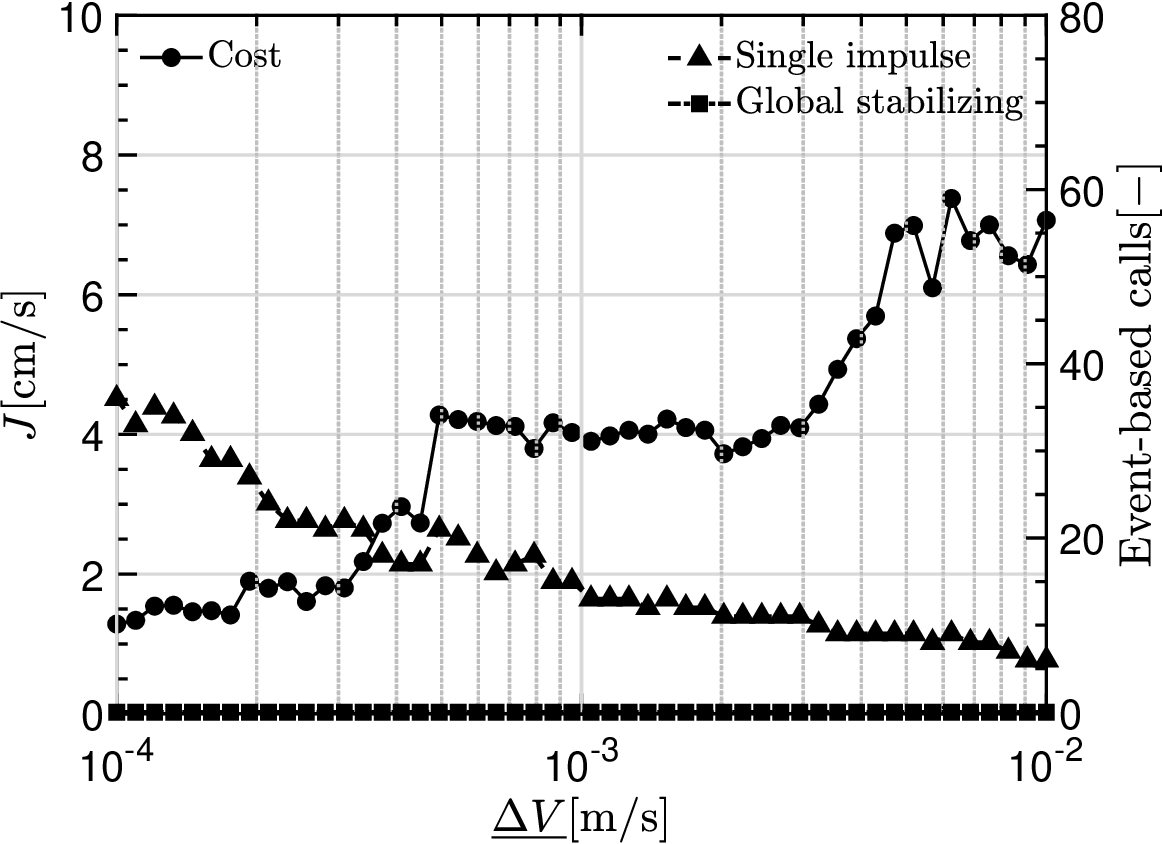}
	\end{center}
	\caption{Cost and control calls of the event-based controller for different dead-zone values.}	
	\label{fig:dVmin_J_calls}
\end{figure}
\begin{figure}[] 
	\begin{center}
		\includegraphics[width=9.5cm,height=9.5cm,keepaspectratio]{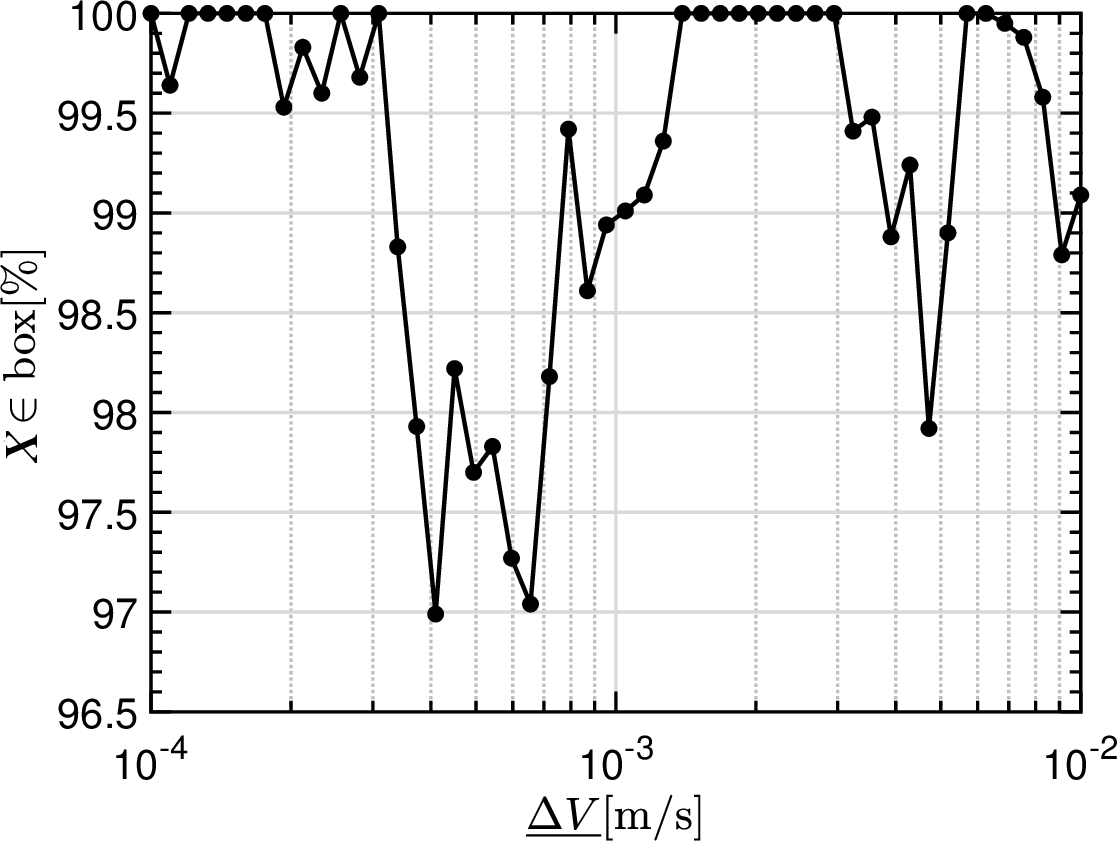}
	\end{center}
	\caption{Polytopic constraints satisfaction time percentage of the event-based controller for different dead-zone values.}	
	\label{fig:dVmin_poly}
\end{figure}
Figure \ref{fig:dVmin_J_calls} shows that the hovering fuel consumption increases gradually from $1.5~\text{cm/s}$ to $3~\text{cm/s}$ when the dead-zone threshold reaches $0.45~\text{mm/s}$. Above this value, a quick raise to $4~\text{cm/s}$ happens. This consumption stays steady until the dead-zone threshold reaches $3~\text{mm/s}$ from where it increases gradually to $7~\text{cm/s}$. In a general way, the fuel consumption trend is to increase as the dead-zone is enlarged. Another trend showed by Fig.\ref{fig:dVmin_J_calls} is decrease of the impulses number as $\underline{\Delta V}$ increases. This highlights a reduction of triggering opportunities when the dead-zone band is higher. In other terms, the event-based control algorithm is less sensitive and reactive as the dead-zone band is enlarged. However in the proposed study, the reactivity is not directly correlated with the accuracy of the control scheme. In fact, \figref{fig:dVmin_poly} shows that the hovering zone constraint satisfaction time percentage is  globally high with values above $97\%$ (with several ones of $100\%$). It is also seen, in Fig.\ref{fig:dVmin_J_calls}, that the single impulse strategy is always commanded at all cases. It can be concluded that the behavior of the event-based controller could be affected in terms of reactivity and consumption if the dead-zone band is enlarged but the control performances remain reasonably good. Note that the dead-zone value is an intrinsic property of the chosen spacecraft thrusters for the mission. Nonetheless, a different value can be set in the control algorithm to tune the behavior of the controller and thrusters.

\subsection{Impact of the saturation}\label{Sat_results}

Finally, for a given dead-zone threshold, $\underline{\Delta V}=10^{-3}~\textup{m/s}$, and eccentricity, $e=0.004$, the impact of the saturation value $\overline{\Delta V}$ is analyzed. To this end, 50 saturation values logarithmically equispaced between $10^{-3}~\textup{m/s}$ and $10^{-1}~\textup{m/s}$ are evaluated.

\begin{figure}[] 
	\begin{center}
		\includegraphics[width=9.5cm,height=9.5cm,keepaspectratio]{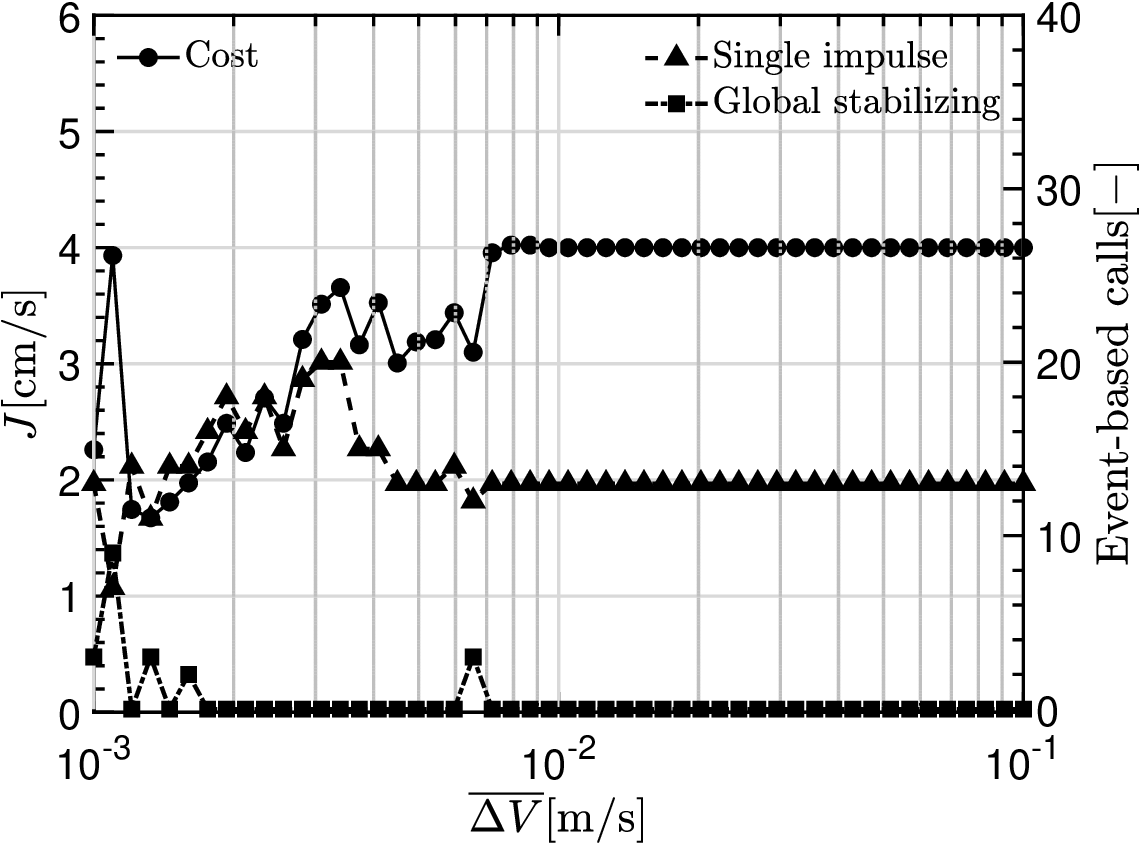}
	\end{center}
	\caption{Cost and control calls of the event-based controller for different saturation values.}	
	\label{fig:dVmax_J_calls}
\end{figure}
\begin{figure}[] 
	\begin{center}
		\includegraphics[width=9.5cm,height=9.5cm,keepaspectratio]{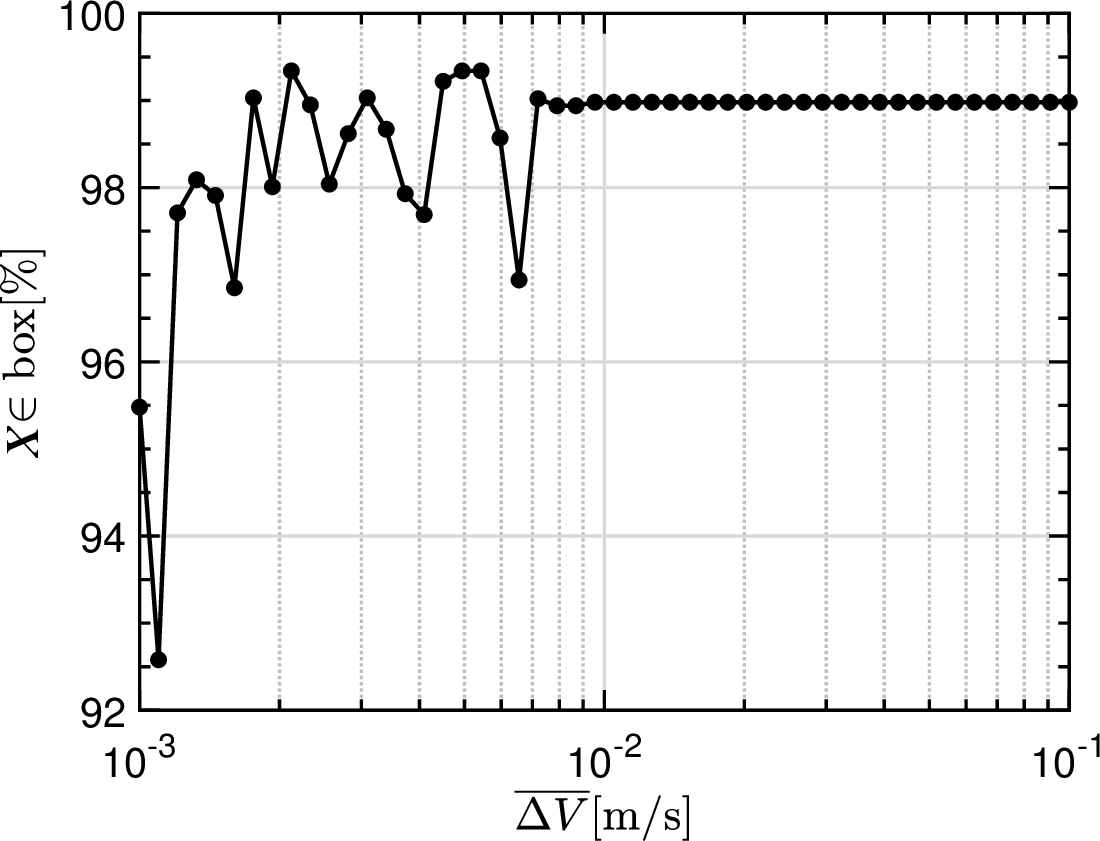}
	\end{center}
	\caption{Polytopic constraints satisfaction time percentage of the event-based controller for different saturation values.}	
	\label{fig:dVmax_poly}
\end{figure}

In Fig.\ref{fig:dVmax_J_calls}, two different trends can be observed. Firstly, when the saturation value is lower than $1~\text{cm/s}$, the number of event-based calls increases as the saturation band enlarges. When the saturation value is closer to the dead-zone value, the global controller is called frequently due to the lesser available control region. The fuel consumption follows exactly the same trend. It is also observed, in \figref{fig:dVmax_poly}, a correlation between global controller calls and the loss of control accuracy. Additionally, the accuracy is maximal and steady when the saturation value enlarges the control region enough.

If the saturation value is greater than 1 cm/s, the fuel consumption, number of control calls and accuracy remain steady. It is concluded that the saturation does not affect the behavior and control performances above a given value. This fact can help to choose the thruster with respect to its capabilities.

\section{Conclusions} 

In this paper, an event-based predictive controller for the spacecraft rendezvous hovering phase has been presented. The event-based control architecture is composed of trigger rules calling a suitable control law. The trigger rules are based on the monitoring the state and the admissible set reachability and the preferred control law is a single impulse approach. 
For instance, it has been highlighted that the resulting event-based algorithm has a high computational efficiency. In fact, the computations involve polynomial roots computation (for the instantaneous reachability) and the cost evaluation in a small number admissible solutions of a given small dimensional linear program. This fact makes the proposed controller suitable for the embedness on spacecraft computational devices. 

The properties single impulse control approach have been assessed in Section \ref{invariance_single_impulse_approach},
The well-posedness of the proposed controller was demonstrated. Moreover, under nominal scenario parameters the single impulse strategy, along with the associated trigger rules, makes the admissible set attractive (even in the presence of weak enough disturbances). This fact has also been remarked under numerical simulations with dynamical disturbances. These simulations emphasize the superiority of the event-based controller performances in terms of accuracy and computational efficiency (with similar fuel consumption and number of relevant impulses) when compared to the global stabilizing method of \cite{Arantes2019}. Finally, it has been highlighted that the thruster capabilities, via dead-zone and saturation parameters, may have an impact in the mission performance.

An extension of this work could possibly be to robustify the control and trigger rules with respect to the presence of impulses mishaps. To this end, a robust formulation with disturbance estimation as \cite{Gavilan2012} combined with state uncertainty minimization, as in \cite{Louembet2015}, could be considered. Finally, another possible future work is the explicit consideration of Earth's oblateness within our formulation by adding the periocidity condition given by \cite{Roscoe2011}.

\section*{Acknowledgments}

The authors gratefully acknowledge financial support from University of Seville, through its V-PPI US, under grant PP2016-6975 and from the Spanish Ministry of Science and Innovation under grant PGC2018-100680-B-C21.

\bibliography{event_control_bibliography}

\appendix
\section{Appendix}

The aim of this appendix is to expose the geometry of the 2-$\pi$ period reachable sets. These sets are relevant to describe the region of attraction $\mathcal D$ and thereafter the so-called dead-zone set, $\mathcal D_{\text{dz}}$. 
Moreover, conditions of existence of the dead-zone set are provided. These conditions depend mainly to the dead-zone threshold $\underline{\Delta V}$.

\subsection{Reachable set over one period} \label{ReachableSet}

Following \cite{Sanchez2019_bis}, it is convenient to define the reachable set over a 2-$\pi$ period $\mathcal{F}=\mathcal{F}_{xz}\times\mathcal{F}_y$. The reachable set is expressed with respect to the state increment $\Delta D=D^+-D\in\R^6$. 

\subsubsection{Out-of-plane}

The out-of-plane state increment is defined as $\Delta D_y=D^+_y-D_y=[\Delta d_4,~\Delta d_5]^T$ 
\begin{equation}
\Delta D_y(\nu,\lambda_y)=\lambda_yB_{D,y}(\nu)=\frac{\lambda_y}{k^2\rho}[
-s_{\nu},~c_{\nu}]^T. \label{eq:inc_outplane}
\end{equation}
Exploiting \cite{Hong1995}, an implicite form of \eqref{eq:inc_outplane} with respect to the independent variable is obtained:
\begin{equation}\label{eq:out_control_ellipse}
f_{y}(\Delta D, \lambda_y)=\dfrac{\Delta d_4^2}{\left(\dfrac{\lambda_y}{k^2\sqrt{1-e^2}}\right)^2}+\dfrac{\left(\Delta d_5+\dfrac{e \lambda_y}{k^2(1-e^2)}\right)^2}{\left(\dfrac{\lambda_y}{k^2(1-e^2)}\right)^2}-1=0.
\end{equation}
Equation \eqref{eq:out_control_ellipse} describes the equation of an ellipse for given fixed parameters $e,~\lambda_y$.
Define the out-of-plane state increment reachable set as
\begin{equation}
\mathcal{F}_y(\Lambda):=\{\Delta D_y \in \R^2: f_{y}(\Delta D_y, \lambda_y)=0,~\forall \lambda_y \in \Lambda\subseteq \R \},
\end{equation}
where the set $\Lambda$ refers to the allowable values for the out-of-plane control, $\lambda_y$. If both dead-zone and saturation constraints are taken into account (see \eqref{eq:Isat_y}), then the reachable set is fonction of the interval set $\Lambda_{\text{sat},y}$
\begin{equation}\label{eq:incset_reachable_y}
\mathcal{F}_{\text{sat},y}=\mathcal{F}_y(\Lambda_{\text{sat},y}).
\end{equation}
In a similar way it is useful to describe the dead-zone reachable set, states increments with a control below the dead-zone threshold $\underline{\Delta V}$, and the unconstrained reachable set as respectively
\begin{equation}
\mathcal{F}_{\text{dz},y}=\mathcal{F}_{y}([-\underline{\Delta V},~\underline{\Delta V}]),~~\text{and}~~ \mathcal{F}^{\infty}_y=\mathcal{F}_y(\R).
\end{equation}
Remark that $\mathcal{F}^{\infty}_y$ is the whole $D_y$ space, $\R^2$.

\subsubsection{In-plane}

The in-plane state increment is defined as $\Delta D_{xz}=D^+_{xz}-D_{xz}=[\Delta d_0,~\Delta d_1,~\Delta d_2,~\Delta d_3]^T$. 
As the periodicity tracking strategy is applied (see \eqref{eq:Dxz+}), then the in-plane state increment is given by
\begin{equation}
\Delta D_{xz}(\nu,\lambda_{xz},D_{xz})=B_{D,xz}(\nu)\left(\lambda_{xz}B^\bot_{d_0,xz}(\nu)+\Delta V^{0}_{xz}(d_0,\nu)\right),\label{eq:inc_inplane}
\end{equation}
where it should be noted that $\Delta D_{xz}$ depends on the actual state due to $d_0$. As a consequence, the 2-$\pi$ period in-plane reachable state (in the $\Delta D_{xz}$ space) is
\begin{equation}
\mathcal{F}_{xz}(\Lambda) :=\{\Delta D_{xz} \in \R^4: \Delta D_{xz}~\text{s.t.}~\text{Eq.}\eqref{eq:inc_inplane},~\forall \lambda_{xz} \in \Lambda\subseteq\R\}.
\end{equation}
$\Lambda$ is again a interval set of the allowable values for the in-plane control variable, $\lambda_{xz}$. To account for dead-zone and saturation constraints, interval set $\Lambda$ is given  by \eqref{eq:Isat_xz}. Consequently, the constrained reachable set is 
\begin{equation}\label{eq:incset_reachable_xz}
\mathcal{F}_{\text{sat},xz}=\mathcal{F}_{xz}(\Lambda_{\text{sat},xz}(d_0,\nu)).
\end{equation}
The dead-zone and unconstrained reachable sets are also defined respectively as
\begin{equation}\label{eq:incset_dead_reachable_xz}
\mathcal{F}_{\text{dz},xz}=\mathcal{F}_{xz}([\underline{\lambda}_{xz,1}(d_0,\nu),~\underline{\lambda}_{xz,2}(d_0,\nu)]),\quad \mathcal{F}^{\infty}_{xz}=\mathcal{F}_{xz}(\R).
\end{equation}
Note that $\Delta d_0=-d_0$ due to the periodicity tracking strategy.
This fact makes the $d_1d_2d_3$ space the relevant one when applying the in-plane impulse. 
A further geometric description can be made under the quasi-steady assumption. 
Assuming that $|d_0|\approx0$, the part of the control impulse that compensate for $d_0$ can be neglected as $\lVert\lambda_{xz}B^\bot_{d_0,xz}\rVert_2\gg\lVert\Delta V^{0}_{xz}\rVert_2$, then
\begin{equation}\label{eq:incset_reachable_xz1}
\Delta D_{xz}(\nu,\lambda_{xz})\approx \lambda_{xz}B_{D,xz}(\nu)B^\bot_{d_0,xz}(\nu)=\frac{\lambda_{xz}}{k^2\rho}[0,~-s_{\nu},~c_{\nu},~2+ec_{\nu}]^T.
\end{equation}
\eqref{eq:incset_reachable_xz1} can be implicitized with respect to $\nu$ and $\lambda_{xz}$ by means of a Groebner basis (see \cite{Fix1995}) to obtain the implicit equation that describe the state increment surface:
\begin{equation}
f_{xz}(\Delta D_{xz})=4\Delta d^2_1+(4-e^2)\Delta d_2^2+2e\Delta d_2\Delta d_3-d_3^2=0,\label{eq:incd1d2d3_cone}.
\end{equation} 
Equation \eqref{eq:incd1d2d3_cone} represents a conic surface in the $\Delta d_1 \Delta d_2 \Delta d_3$ space being the $\Delta d_3$ axis the apex if $e=0$. 
It provides an approximation for the in-plane unconstrained reachable set over a 2-$\pi$ period,
\begin{equation}
\mathcal{F}^{\infty}_{xz}\approx(\Delta D_{xz}\in\R^4:f_{xz}(\Delta D_{xz})=0).
\end{equation}
Since the control variable $\lambda_{xz}$ was lost due to implicitization, this analytical description can not be extended to account for dead-zone and saturation constraints. But, as matter of fact, $\mathcal{F}_{\text{sat},xz}$ and $\mathcal{F}_{\text{dz},xz}$ are sections of the $\mathcal{F}^{\infty}_{xz}$.

\subsection{Dead-zone set}\label{DeadzoneSet}

The dead-zone set $\mathcal{D}_{\text{dz}}=\mathcal{D}_{\text{dz},xz}\times\mathcal{D}_{\text{dz},y}$ is defined as the set of states from where all the $\adset$ reachability opportunities over a 2-$\pi$ period fall within the dead-zone threshold. 
This set is of particular because, whenever $D\in\mathcal{D}_{\text{dz}}$, the global controller is called (see Algorithm \ref{alg1}). As it is shown in the sequel, this set may or may not exist depending on the conditions developed thereafter.

\subsubsection{Out-of-plane}

Using a formal notation, the out-of-plane dead-zone set is defined as
\begin{eqnarray}
\mathcal{D}_{\text{dz},y}&:=&\left\{D_y\in\R^2 : D_y\notin \adsety,~~D_y\notin\mathcal{D}_y,~~ (D_y\oplus\mathcal F_{\text{dz},y})\cap\adsety \neq \emptyset\right\},\\
&:=&\left\{ D_y\in\R^2: D_y\notin\adsety,~~ (D_y\oplus \mathcal F_{\text{sat},y})\cap \adsety = \emptyset,~~  (D_y\oplus\mathcal F_{\text{dz},y}\cap\adsety \neq \emptyset  \right\}.
\end{eqnarray}
Following \cite{Sanchez2019_bis}, the contractive set $\mathcal{D}_y$ can be expressed in terms of the following Minkowski sum
\begin{equation}\label{eq:cont_set_outplane_Minkowski}
\mathcal{D}_y=\adsety\oplus\mathcal{F}_{\text{sat},y}.
\end{equation}
For simplicity, the summation has been considered since the out-of-plane state increment reachable set over one period, see Eq.\eqref{eq:incset_reachable_y}, does not depend on the current state $D_y$. Since $\mathcal{F}_{\text{sat},y}$ is the covered region between two ellipses (which is closed), it is deduced that the out-of-plane dead-zone set only exist, $\mathcal{D}_{\text{dz},y}\neq\emptyset$, if $\adsety\subset\mathcal{F}_{\text{dz},y}$ by convexity of Eq\eqref{eq:cont_set_outplane_Minkowski} Minkowski sum. This is an important conclusion since only the minimum impulse bit (which depends on the thruster capabilities) could degrade the out-of-plane event-based controller behavior. For example, consider the case where $e=0$, then both $\adsety$ and the interior region of $\mathcal{F}_y$ are circles of radius $\text{max}\{|\ymin|,|\ymax|\}$ and $\underline{\Delta V}/k^2$ respectively. It can be easily seen that $\adsety\subset\mathcal{F}_{\text{dz},y}$ if $\underline{\Delta V}> k^2\text{max}\{|\ymin|,|\ymax|\}$. Figure \ref{fig:state_outplane_reachability} shows a case where the out-of-plane state is within the contractive set and another case where the out-of-plane state is within the dead-zone set (thus, $\adsety$ is unreachable over the 2-$\pi$ period). Figure \ref{fig:out_of_plane_invariance}(a), presents the nominal scenario (simulated in Section \ref{Results}) where the dead-zone set vanishes ($\mathcal{D}_{\text{dz},y}=\emptyset$). Figure \ref{fig:out_of_plane_invariance}(b) shows a case where the dead-zone value is augmented considerably, thus the dead-zone set, $\mathcal{D}_{\text{dz},y}$ has a relevant size.

\begin{figure}[] 
	\begin{center}
		\subfigure[]{\includegraphics[width=8cm,height=8cm,keepaspectratio]{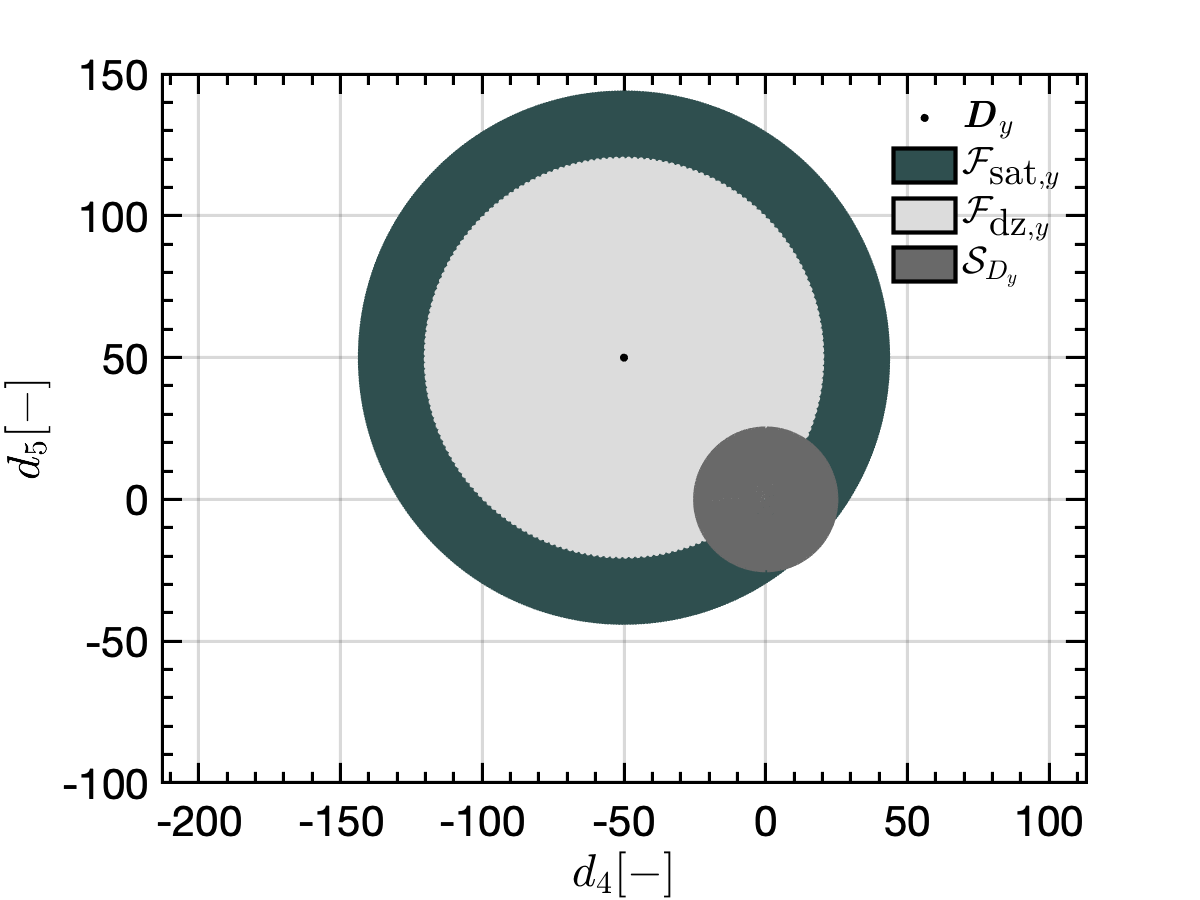} \label{subfig:state_outplane_mathcalD}}
		\subfigure[]{\includegraphics[width=8cm,height=8cm,keepaspectratio]{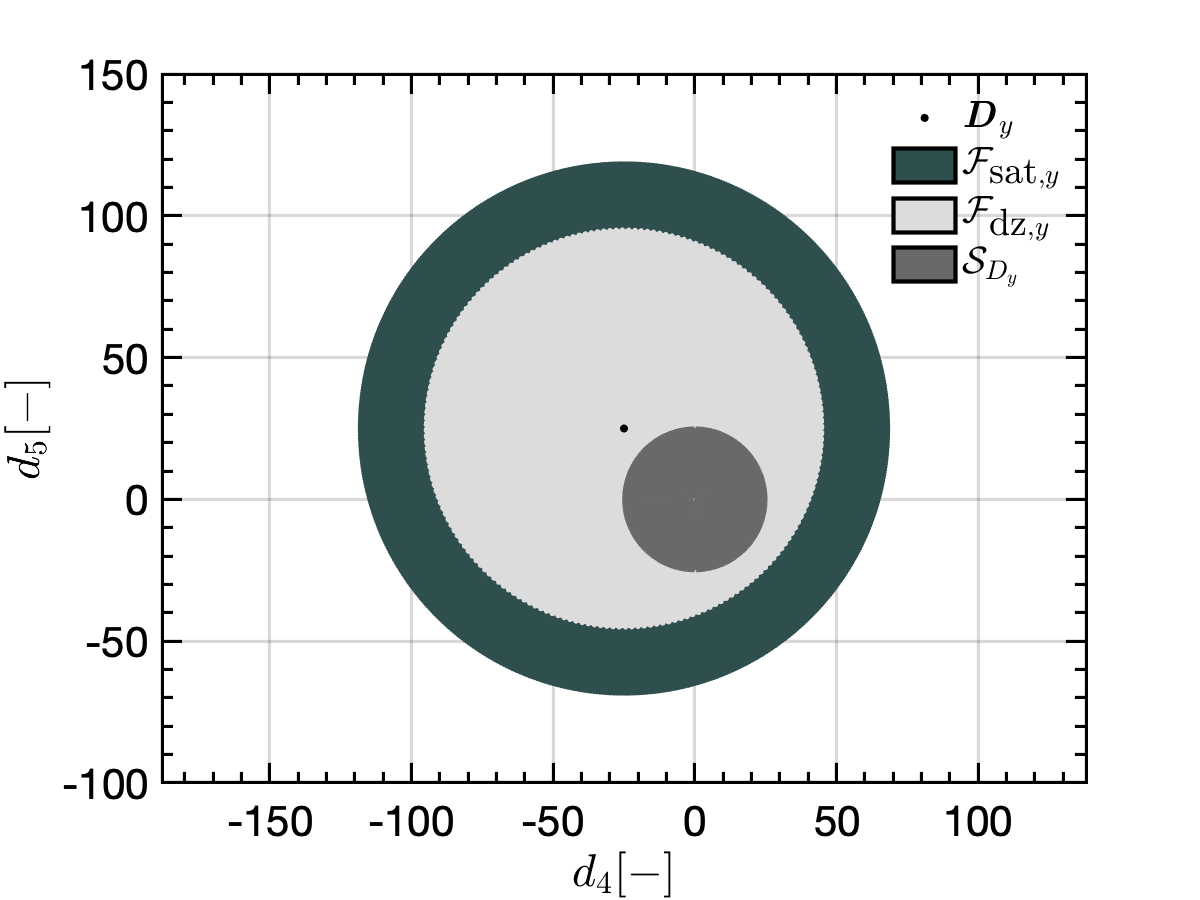}\label{subfig:state_outplane_dead-zone}}\qquad
		\caption{$\mathcal{F}_{\text{sat},y}$, $\mathcal{F}_{\text{dz},y}$ and $\adsety$ for $a=7011~\text{km}$ $e=0.004$, $\ymax=-\ymin=25~\text{m}$, $\underline{\Delta V}=7.5\cdot10^{-2}~\text{m/s}$, $\overline{\Delta V}=0.1~\text{m/s}$, (a): $D_y=[-50,~50]^T\in\mathcal{D}_y$, (b): $D_y=[-25,~25]^T\in\mathcal{D}_{\text{dz},y}$.}
		\label{fig:state_outplane_reachability}
	\end{center}
\end{figure}

\begin{figure}[] 
	\begin{center}
	\subfigure[]{\includegraphics[width=8cm,height=8cm,keepaspectratio]{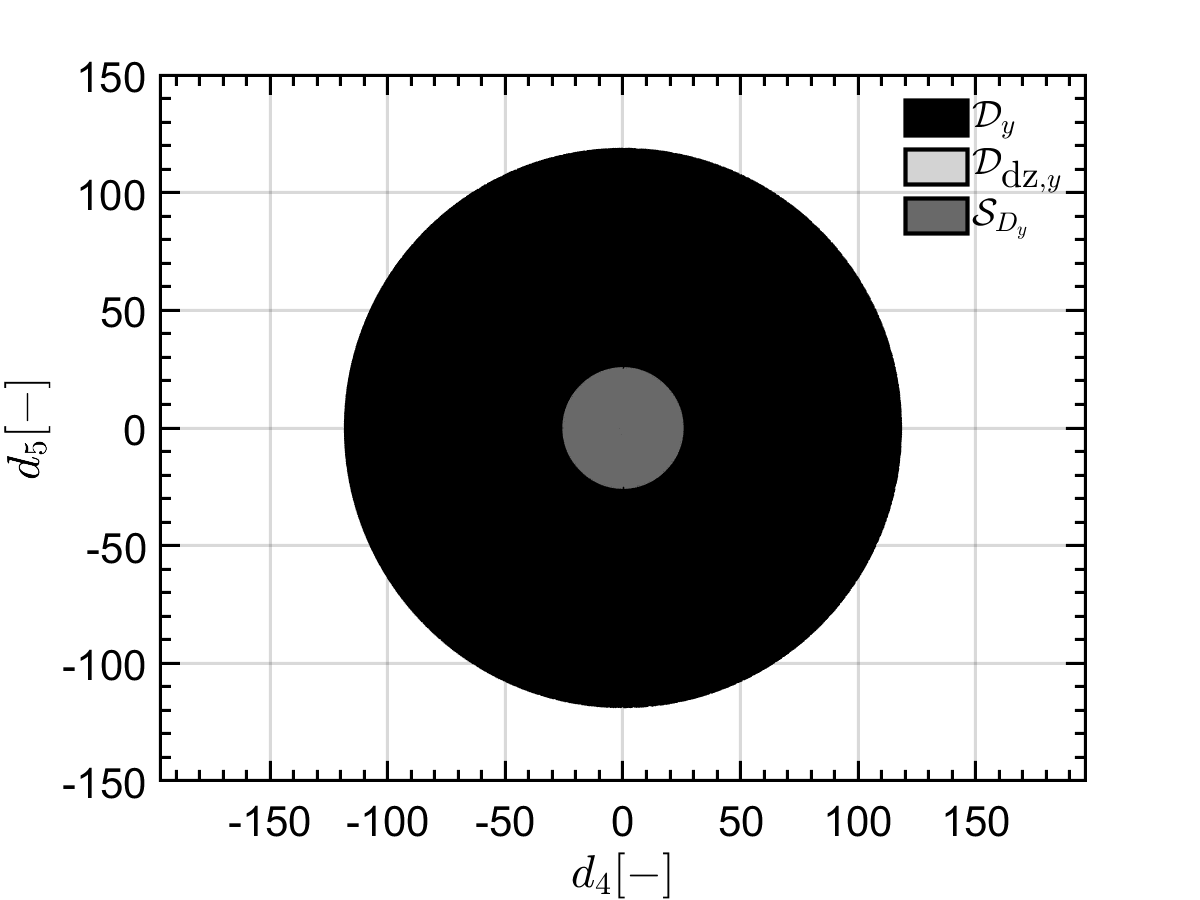} \label{subfig:nominal_out_of_plane}}
	\subfigure[]{\includegraphics[width=8cm,height=8cm,keepaspectratio]{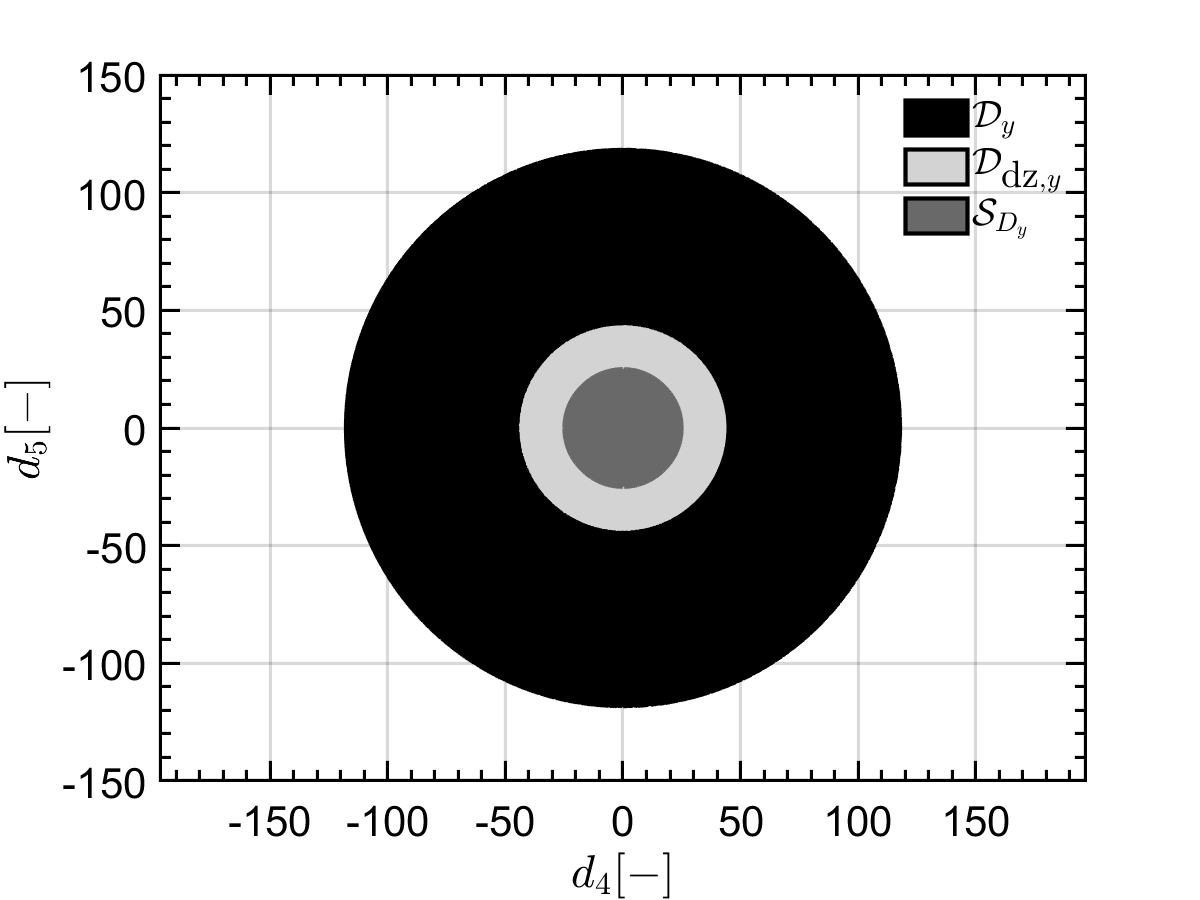} \label{subfig:extreme_out_of_plane}}
	\caption{$\mathcal{D}_{y}$, $\mathcal{D}_{\text{dz},y}$ and $\adsety$  for $a=7011~\text{km}$ $e=0.004$, $\ymax=-\ymin=25~\text{m}$, $\overline{\Delta V}=0.1~\text{m/s}$, (a): $\underline{\Delta V}=10^{-3}~\text{m/s}$, (b): $\underline{\Delta V}=7.5\cdot10^{-2}~\text{m/s}$.}
	\label{fig:out_of_plane_invariance}
	\end{center}
\end{figure}

\subsubsection{In-plane}

Using the previous notation, the in-plane dead-zone set is defined as
\begin{eqnarray}
\mathcal{D}_{\text{dz},xz}&:=&\left\{D_{xz}\in\R^2 : D_{xz}\notin \adsetxz,~~D_{xz}\notin\mathcal{D}_{xz},~~ (D_{xz}\oplus\mathcal F_{\text{dz},xz})\cap\adsetxz \neq \emptyset\right\},\\
&:=&\left\{D_{xz}\in\R^2: D_{xz}\notin\adsetxz,~~ (D_{xz}\oplus \mathcal F_{\text{sat},xz})\cap \adsetxz = \emptyset,~~  (D_{xz}\oplus\mathcal F_{\text{dz},xz}\cap\adsetxz \neq \emptyset\right\}.
\end{eqnarray}
Under the quasi-steady assumption, the in-plane contractive set $\mathcal{D}_{xz}$ can be approximated as
\begin{equation}
\mathcal{D}_{xz}\approx\adsetxz\oplus\mathcal{F}_{\text{sat},xz}~~\text{if}~~|d_0|\approx0.
\end{equation}
Note that the in-plane state increment reachable set independent of the current state. 
However, in this case, the Minkowski sum is composed of a convex closed set, $\adsetxz$, with some sections of the conic surface (due to dead-zone and saturation) given by Eq.\eqref{eq:incd1d2d3_cone}. Since a cone is an open surface, no conclusions can be yielded about possible event-based controller degradation causes. This degradation could be caused due to a combination of the dead-zone threshold and the problem topology. Figure \ref{fig:state_inplane_reachability} shows a case where the in-plane state is within the in-plane contractive set $\mathcal{D}_{xz}$ (\figref{subfig:state_inplane_mathcalD}) and another situation where the state is within the in-plane dead-zone set $\mathcal{D}_{\text{dz},xz}$ (\figref{subfig:state_inplane_dead-zone}). Figure \ref{fig:in_plane_invariance} show a case where the dead-zone does not exist \figref{subfig:in_plane_invariance_satisfaction} and another where it exists \figref{subfig:in_plane_invariance_violation}.

\begin{figure}[] 
	\begin{center}
		\subfigure[]{\includegraphics[width=8cm,height=8cm,keepaspectratio]{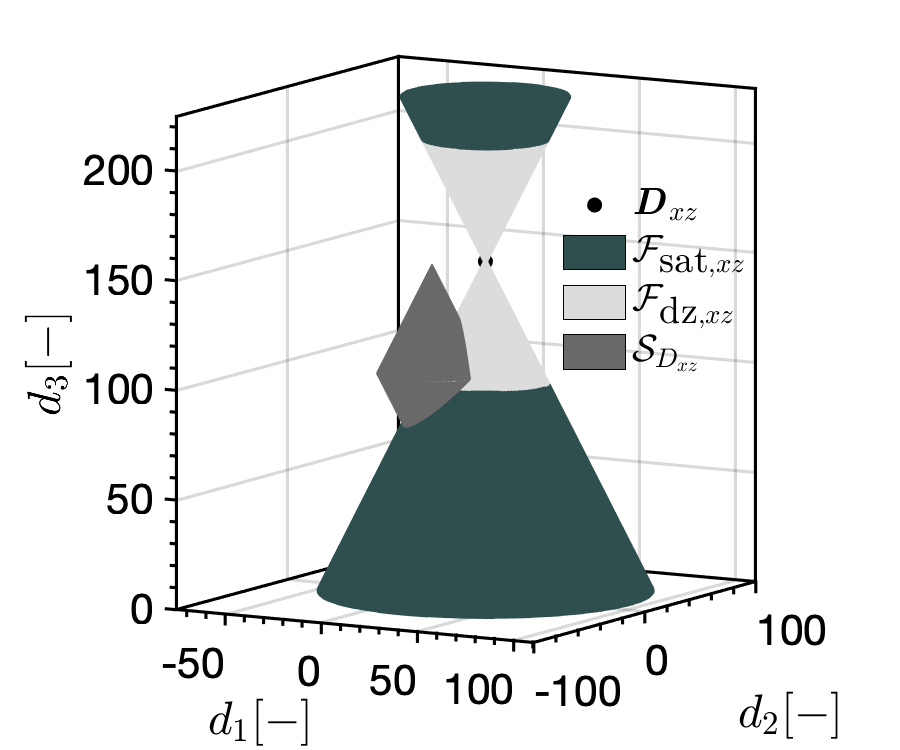} \label{subfig:state_inplane_mathcalD}}
		\subfigure[]{\includegraphics[width=8cm,height=8cm,keepaspectratio]{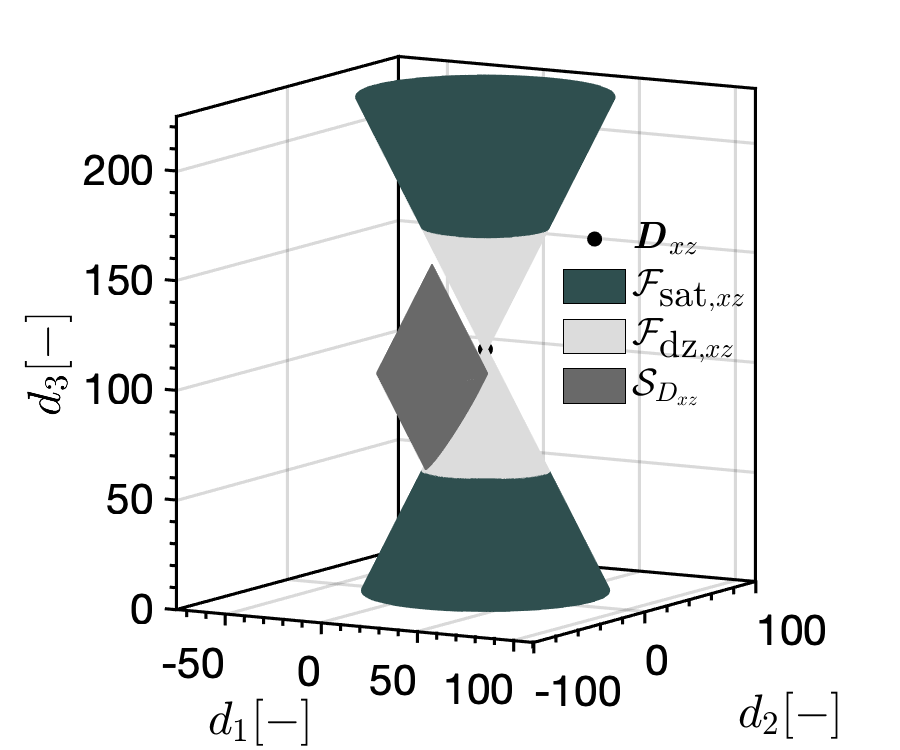}\label{subfig:state_inplane_dead-zone}}\qquad
		\caption{$\mathcal{F}_{\text{sat},xz}$, $\mathcal{F}_{\text{dz},xz}$ and $\adsetxz$ for $a=7011~\text{km}$, $e=0.004$, $\ymax=-\ymin=25~\text{m}$, $\underline{\Delta V}=3\cdot10^{-2}~\text{m/s}$, $\overline{\Delta V}=0.1~\text{m/s}$, (a): $D_{xz}=[0,~17.5,~17.5,~150]^T\in\mathcal{D}_{xz}$, (b): $D_{xz}=[0,~17.5,~17.5,~110]^T\in\mathcal{D}_{\text{dz},xz}$.}
		\label{fig:state_inplane_reachability}
	\end{center}
\end{figure}

\begin{figure}[] 
	\begin{center}
		\subfigure[]{\includegraphics[width=8cm,height=8cm,keepaspectratio]{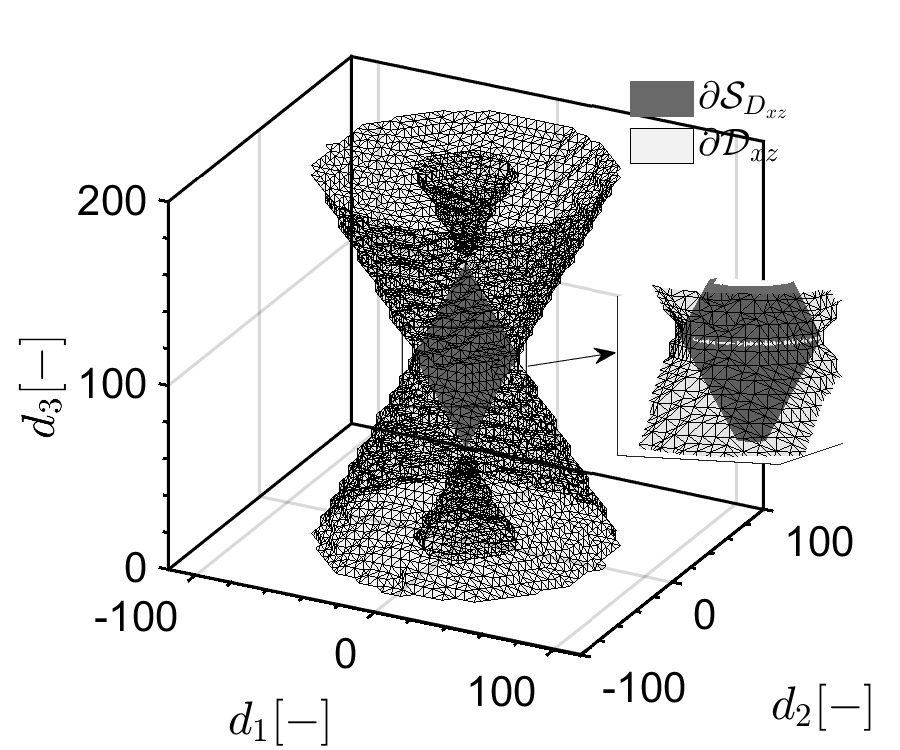} \label{subfig:in_plane_invariance_satisfaction}}
		\subfigure[]{\includegraphics[width=8cm,height=8cm,keepaspectratio]{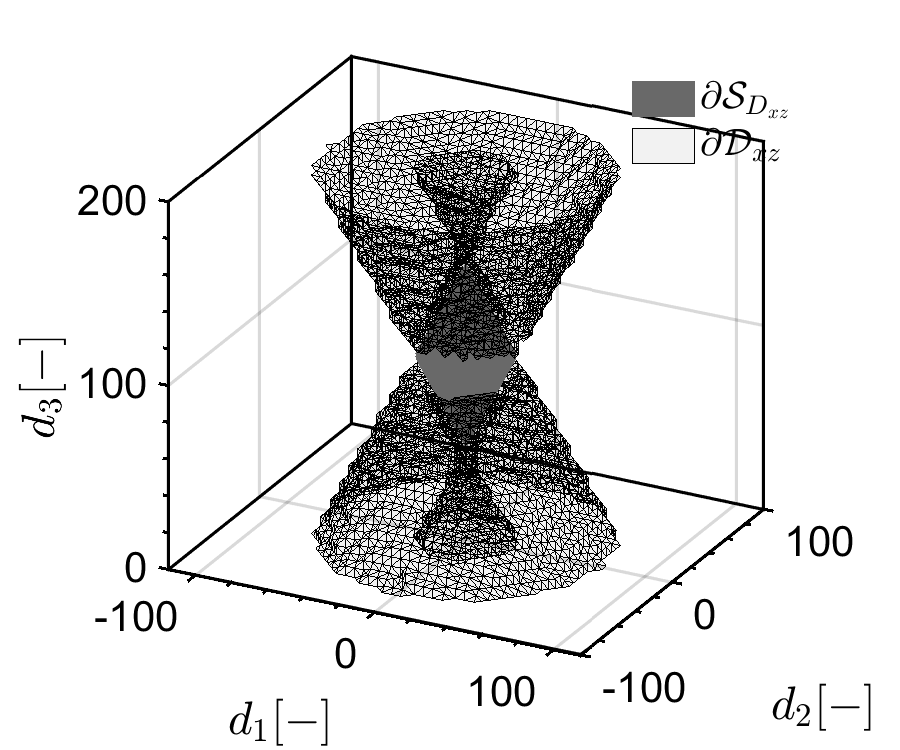} \label{subfig:in_plane_invariance_violation}}
		\caption{$\partial\mathcal{D}_{xz}$ and $\partial\adsetxz$ for $a=7011~\text{km}$ $e=0.004$, $\xmax=150~\text{m}$, $\xmin=50~\text{m}$, $\zmax=-\zmin=25~\text{m}$, $\overline{\Delta V}=0.1~\text{m/s}$, (a): $\overline{\Delta V}=10^{-3}~\text{m/s}$, (b): $\overline{\Delta V}=3\cdot10^{-2}~\text{m/s}$.}
		\label{fig:in_plane_invariance}
	\end{center}
\end{figure} 

\end{document}